\documentclass[journal,twoside,web]{ieeecolor}
\usepackage{generic}
\usepackage{cite}
\usepackage{amsmath,amssymb,amsfonts}
\usepackage{algorithmic}
\usepackage{graphicx}
\usepackage{textcomp}
\usepackage{subcaption}
\usepackage{mathtools}
\usepackage{titlesec}
\titlespacing*{\section}{0pt}{8pt}{2pt}
\titlespacing*{\subsection}{0pt}{6pt}{1pt}
\titlespacing*{\subsubsection}{0pt}{4pt}{1pt}
\newcommand{\nvec}[1]{\mathbf{n}(#1)}
\newcommand{\Rot}[1]{\mathbf{R}(#1)}
\newtheorem{lemma}{\textbf{Lemma}}  
\newtheorem{remark}{\textbf{Remark}}
\newtheorem{thm}{\textbf{Theorem}}
\newtheorem{assumption}{\textbf{Assumption}}

\newtheorem{corollary}{\textbf{Corollary}}
\newtheorem{Proposition}{\textbf{Proposition}}
\newtheorem{problemx}{\textbf{Problem}}
\newenvironment{problem}{%
  \problemx
}{\endproblemx}
\newcommand{\sgn}{\operatorname{sgn}}

\newcommand{\dist}{\operatorname{dist}}
\newcommand{\atantwo}{\operatorname{atan2}}

\def\BibTeX{{\rm B\kern-.05em{\sc i\kern-.025em b}\kern-.08em
    T\kern-.1667em\lower.7ex\hbox{E}\kern-.125emX}}
\begin{document}
\title{Bearing-based Circumnavigation with Collision Avoidance in Time-varying Graphs under Limited Target Information}
\author{Kushal Pratap Singh$^{1}$, \IEEEmembership{Member, IEEE}, Twinkle Tripathy$^{2}$, \IEEEmembership{Senior Member, IEEE} and Anoop Jain$^{3}$, \IEEEmembership{Senior Member, IEEE}
\thanks{*This work was not supported by any organisation.}
\thanks{$^{1}$Kushal P. Singh is a research scholar and $^{2}$Twinkle Tripathy is an assistant professor with the department of Electrical Engineering, Indian Institute of Technology Kanpur,
        Kanpur-208016, India
        {\tt\small kushalp20@iitk.ac.in}, {\tt\small ttripathy@iitk.ac.in}}
\thanks{$^{3}$Anoop Jain is an assistant professor with the Department of Electrical Engineering, Indian Institute of Technology Jodhpur,
        Jodhpur-342030, India
        {\tt\small anoopj@iitj.ac.in}}}

\maketitle

\begin{abstract}
In this paper, we study the distributed circumnavigation of a stationary target by unicycle agents with possibly different constant linear speeds. Each agent is modelled as a disk rather than a point mass to account for its physical dimensions. The target location is assumed to be accessible only to a small subset of agents, called leaders. The rest, called followers, therefore use only local information available from their designated out-neighbour in the interaction graph characterised by the selection of nearest neighbours. By controlling only angular speeds, we develop a distributed guidance law to circumnavigate a stationary target. The proposed guidance law works for both static and time-varying interaction graphs. A logarithmic Barrier Lyapunov Function (BLF) is imposed on the relative heading error, and inter-agent collision avoidance is enforced through forward invariance of compact sublevel sets of the resulting Lyapunov function together with an explicit design condition on the formation parameter. We show that every follower converges to circumnavigation about the same target as the leader at the end of its directed path in the interaction graph, provided the initial conditions are admissible. Numerical simulations illustrate the effectiveness of the proposed method for both static and time-varying topologies.
\end{abstract}


\section{Introduction}
\label{Section:Introduction}
The problem of bearing-based circumnavigation pertains to steering one or more agents to move on circular paths around a target using bearing (angle) measurements. It has emerged as a critical task in a wide range of application domains, including boundary surveillance, asset protection, border patrol, search assist, hazardous rescue missions etc.~\cite{matveev2011method,fu2023justification,cooper2020optimal,leonard2007collective,lytridis2021overview,ahmadzadeh2006multi,girard2004border,quigley2005towards}.

A significant portion of the existing literature, such as in~\cite{2022_single,sinha2023}, has focused on a single agent performing circumnavigation of a stationary target. While these approaches provide valuable insights into the fundamental mechanics of circumnavigation, they do not directly address the coordination challenges that arise in multi-agent settings. Cooperative strategies offer several advantages, including improved scalability, robustness against individual-agent failures, and reduced communication overload. Early contributions, such as~\cite{all_to_all_communication,limited_communication,marshall2004formations}, focused on the design of control laws for groups of unicycle agents operating with identical constant linear speeds, which was later generalised for different linear speeds in~\cite{sinha2007generalization}. For agents with constant identical linear speeds, the authors of~\cite{TRIPATHY2024111315} obtain convergence from almost all initial conditions.

Building on these cooperative frameworks, subsequent research has explored a broader range of formation objectives. Examples include circular, elliptical and spiral formations~\cite{ramirez2010distributed}, dynamic unicycles circumnavigating a location based on initial conditions~\cite{el2012distributed}, demonstration of distinct classes of circular motion~\cite{seyboth2014collective} and modulation of both linear and angular speeds~\cite{zheng2015distributed}.
In practice, circumnavigation laws must explicitly account for safety constraints arising from the finite size of agents and their motion. BLF-based methods provide a natural way to incorporate such state constraints into the stability analysis. Existing BLF constructions include re-centred barrier functions~\cite{panagou2015distributed}, parametric barrier formulations~\cite{han2019robust}, universal barrier functions~\cite{jin2021multirobot} etc. In this paper, we adopt a logarithmic BLF~\cite{tee2009barrier} because it leads to a tractable distributed design, renders the region of attraction explicit, and permits a direct proof of inter-agent separation.

In addition to physical constraints, another fundamental challenge arises from limited information availability within the network. In many practical scenarios, access to the target location is restricted to only a small subset of the agents. Under such conditions, the problem formulation depends critically on the underlying communication topology, for instance, directed graphs possessing a spanning tree or cyclic interaction structures. Collectively, this highlights the need for distributed, information-efficient, and safety-critical control strategies for circumnavigation. Moreover, in practice, interaction topologies often evolve over time due to factors such as component failures, obstacles, and limited fields of view~\cite{11004631}. To the best of our knowledge, a unified framework that simultaneously addresses bearing-only sensing by the followers, heterogeneous agents, collision avoidance, limited target information and time-varying interaction topology remains unexplored.

Motivated by these considerations, we investigate cooperative circumnavigation for a group of heterogeneous unicycle agents with fixed safety radii under limited information availability and a time-varying interaction topology. As in~\cite{chan2020angle}, fixed radii model the agents' physical size. The proposed approach employs a bearing-based distributed guidance strategy, in which follower agents rely solely on local angular measurements and information from a single out-neighbour, whereas leaders utilise target-bearing information. A logarithmic BLF is embedded in the follower design; it constrains the relative heading error, renders the admissible domain forward invariant, and makes the region of attraction explicit.
Furthermore, we consider the case of a time-varying interaction topology based on selecting the nearest out-neighbour in the interaction graph. The main contributions of this paper are as follows:
\begin{enumerate}
    \item \textit{Asymptotic convergence with safety guarantees:} Regulating only angular speeds, we develop a distributed guidance law to circumnavigate a stationary target. For static interaction graphs, we establish asymptotic stability of the desired relative equilibrium with a region of attraction equal to the entire admissible domain, together with forward invariance of the collision-free set.
    \item \textit{Equilibrium states based formation control:} We show that when information can flow through the interaction graph, in particular, when every follower has a directed path to a leader, the shape of the formation can be controlled by selecting different equilibrium states within the admissible set. 
    \item \textit{Applicability to time-varying graphs:} Along with static interaction graphs, we analyse piecewise-static time-varying graphs with node-entry and node-exit events.
    \item \textit{Range-free distributed guidance law:} The follower law is a function of the locally measured angles, requires no range-rate measurement, no differentiation of measured signals and no state memory, and is well defined on the whole admissible domain, which makes the design suitable for practical and distributed implementation.
\end{enumerate}

The paper is outlined as follows. Section~\ref{Section:preliminaries} introduces the essential background and preliminaries, and Section \ref{Section:Problem_Formulation} states the model and the problem. Section \ref{sec:guidance_for_leaders} and \ref{sec:follower_guidance} design and analyse the leader and follower guidance law, and Section~\ref{sec:switch} deals with switching graphs. Section~\ref{Sec:Simulaton} presents simulations and Section~\ref{Sec:Conclusion} concludes the paper.

\section{Preliminaries}
\label{Section:preliminaries}
\textit{Notations:} $\mathbb R$, $\mathbb R_{\ge 0}$, $\mathbb R_{>0}$ and $\mathbb N$ denote the real, nonnegative real, positive real and natural numbers; $\mathbb R^n$ is the set of $n$-dimensional real vectors, $\|\bullet\|$ the Euclidean norm, and $\langle x,y\rangle:=x^\top y$. $\mathbb S^1$ is the unit
circle, $\operatorname{sgn}(\cdot)$ the signum function, and $\operatorname{atan2}(y,x)\in(-\pi,\pi]$ the four-quadrant arctangent. $\iota$ denotes the imaginary number. For $\theta\in\mathbb{R}$ we write $\nvec{\theta}:=[\cos\theta\ \ \sin\theta]^\top$ for the unit vector of heading $\theta$, and $\Rot{\mu}$ for the planar rotation matrix by angle $\mu$, so that $\Rot{-\gamma}\nvec{\theta}=\nvec{\theta-\gamma}$ and $\langle\nvec{a},\nvec{b}\rangle=\cos(a-b)$. Finally, $J:=\left[\begin{smallmatrix}0&-1\\1&0\end{smallmatrix}\right]$ and we use the identities $\frac{d}{d\theta}\nvec{\theta}=\nvec{\theta+\pi/2}$ and $\frac{d}{d\mu}\Rot{\mu}=\Rot{\mu}J=J\Rot{\mu}$.

\textit{Graph theory:} A directed graph $\mathcal G=(\mathcal V,\mathcal E)$ has node set $\mathcal V=\{1,\dots,n\}$ and edge set $\mathcal E\subseteq\mathcal V\times\mathcal V$. Nodes represent agents and edges interaction between them. If $(i,j)\in\mathcal E$, node $i$ uses information from node $j$: $j$ is an out-neighbour of $i$, and $i$ is an in-neighbour of $j$; $\mathcal N_i^{\rm out}:=\{j:(i,j)\in\mathcal E\}$. A directed path from $i$ to $j$ is a sequence of distinct nodes $i=i_0,i_1,\dots,i_m=j$ with $(i_{k-1},i_k)\in\mathcal E$, $k=1,\dots,m$; its length is $m$. A time-varying graph $\mathcal G(t)=(\mathcal V(t),\mathcal E(t))$ has node and edge sets that may change with $t$.

\textit{Barrier Lyapunov Function (BLF):} A BLF is a scalar function $V(x)$, defined with respect to the system $\dot{x} = f(x)$ on an open region $D$ containing the origin, such that it is continuous, positive definite, and has continuous first-order partial derivatives at every point of $D$. Moreover, it satisfies $V(x) \to \infty$ as $x$ approaches the boundary of $D$, and along the solution of $\dot{x} = f(x)$ for $x(0) \in D$, it holds that $V(x(t)) \leqslant b$ for all $t \geqslant 0$, for some positive constant $b$.

\begin{lemma}[\cite{tee2009barrier}]
\label{Lemma:BLF}
For any positive constants $k_{a_1}$ and $k_{b_1}$, let $Z_1 := \left\{ z_1 \in \mathbb{R} \;:\; -k_{a_1} < z_1 < k_{b_1} \right\} \subset \mathbb{R},$ and $\mathcal{N} := \mathbb{R}^l \times Z_1 \subset \mathbb{R}^{l+1}$ be open sets. Consider the system
\setlength{\abovedisplayskip}{2pt}
\setlength{\belowdisplayskip}{2pt}
\begin{equation}
\dot{\eta} = h(t,\eta)
\end{equation}
where $\eta := [w, z_1]^\top \in \mathcal{N}$, and 
$h : \mathbb{R}_+ \times \mathcal{N} \to \mathbb{R}^{l+1}$ is piecewise continuous in $t$ and locally Lipschitz in $z$, uniformly in $t$, on $\mathbb{R}_+ \times \mathcal{N}$. Suppose that there exist functions $U : \mathbb{R}^l \to \mathbb{R}_+$ and $V_1 : Z_1 \to \mathbb{R}_+$, continuously differentiable and positive definite in their respective domains, such that
\setlength{\abovedisplayskip}{2pt}
\setlength{\belowdisplayskip}{2pt}
\begin{equation}
V_1(z_1) \to \infty \quad \text{as} \quad z_1 \to -k_{a_1} \;\text{or}\; z_1 \to k_{b_1},
\end{equation}
\begin{equation}
\gamma_1(\|w\|) \leqslant U(w) \leqslant \gamma_2(\|w\|),
\end{equation}
where $\gamma_1$ and $\gamma_2$ are class $\mathcal{K}_\infty$ functions. Let $V(\eta) := V_1(z_1) + U(w),$ and $z_1(0)$ belong to the set $z_1 \in (-k_{a_1}, k_{b_1})$. If the following inequality holds:
\setlength{\abovedisplayskip}{2pt}
\setlength{\belowdisplayskip}{2pt}
\begin{equation}
\dot{V} = \frac{\partial V}{\partial \eta} h(t,\eta) \leqslant 0,
\end{equation}
then $z_1(t)$ remains in the open set $(-k_{a_1}, k_{b_1})$ for all $t \in [0,\infty)$.
\end{lemma}

\section{System model and problem formulation}
\label{Section:Problem_Formulation}
Circumnavigation, in which multiple agents maintain circular trajectories around a central target (see Fig. \ref{figure:desired_formation}), plays a crucial role in many operational scenarios. These include surveillance, reconnaissance, environmental monitoring, and coordinated multi-robot firefighting.  In these applications, the stationary target denotes any landmark, beacon or any region of interest. Motivated by these practical applications, we investigate this problem for a group of $n$  heterogeneous autonomous agents governed by the following model.

\subsection{Agent model}
\label{subsec:agent_model}
The target is stationary and, without loss of generality, located at the origin (marked $T$ in the figures). The agents $\mathcal{V} = \{1,...,n\}$ are partitioned into leaders $\mathcal{L} \neq 0$, which know the target position, and followers $\mathcal{F} \coloneqq \mathcal{V}\setminus \mathcal{L}$, which do not. Agent $k$ is a disk of radius $R_s>0$ with centre $P_k=[x_k \quad y_k]^\top$, heading $\gamma_k=\mathbb{S}^1$, constant speed $v_k \in \mathbb{R}_{>0}$ and angular turn rate $u_k$ (see Fig.~\ref{figure:unicycle_model}) with kinematics:
\setlength{\abovedisplayskip}{2pt}
\setlength{\belowdisplayskip}{2pt}
\begin{equation}
\label{eq:unicycle_kinematics}
\dot{x_k} = v_k \cos \gamma_k, \quad
\dot{y_k} = v_k \sin \gamma_k, \quad
\dot{\gamma}_k = u_k,
\end{equation}
\begin{figure}[ht]
    \centering
        \centering    
        \includegraphics[width=0.6\linewidth]{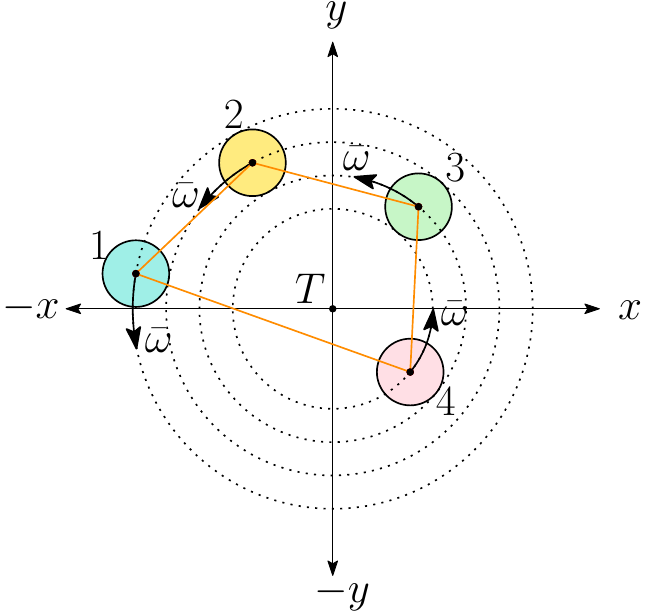}
    \caption{Desired circumnavigation}
    \label{figure:desired_formation}
\end{figure}
\begin{figure}[ht]
\begin{center}
\includegraphics[scale=0.9]{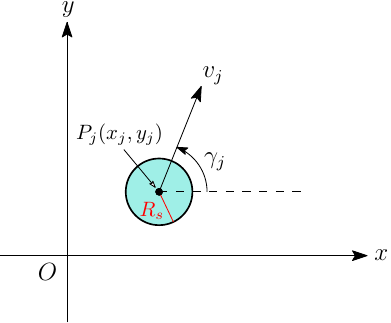}    
\caption{Unicycle with fixed safety radius $R_s$}  
\label{figure:unicycle_model}                                 
\end{center}                                 
\end{figure}

Let $r_k \coloneqq \|P_k\|$, $\lambda_k \coloneqq \operatorname{\atantwo}(y_k,x_k)$, and let $\phi_k \coloneqq \gamma_k-\lambda_k$ be the lead angle, measured from outward radial direction. The goal is circumnavigation at a common angular rate $\bar{\omega}$ ($>0$ for anticlockwise motion, see Fig.~\ref{figure:desired_formation}); at this rate, agent $k$ moves on the circle of radius 
\setlength{\abovedisplayskip}{2pt}
\setlength{\belowdisplayskip}{2pt}
\begin{equation}
    \label{eq:R_k}
    R_k \coloneqq v_k/\omega_k.
\end{equation}
The $\bar{\omega}$-turn centre of each agent k,
\setlength{\abovedisplayskip}{2pt}
\setlength{\belowdisplayskip}{2pt}
\begin{equation}
    \label{eq:turn_centre}
    c_k\coloneqq P_k + R_k \nvec{\gamma_k+\pi/2},
\end{equation}
is the centre of the circle that agent $k$ traces when $u_k \equiv \bar{\omega}$. The collision free set is 
\begin{equation}\label{eq:Xadm}
\mathcal{X}_{\rm adm}:=\big\{(P_k,\gamma_k)_{k\in\mathcal{V}}:\ \|P_i-P_k\|>2R_s\ \ \forall i\neq k\big\}.
\end{equation}

\begin{remark}[Clockwise rotation]\label{rem:cw}
The reflection $(x,y)\mapsto(x,-y)$ maps clockwise circles to anticlockwise ones and reverses the signs of all headings, turn rates and relative angles; hence the clockwise case follows by applying the anticlockwise design in reflected coordinates.
\end{remark}

Now, we introduce the following sensing model to characterise the information exchange among agents.

\subsection{Sensing model}
\label{subsec:sensing_model}
\begin{figure}[ht]
    \centering
    \includegraphics[scale=0.7]{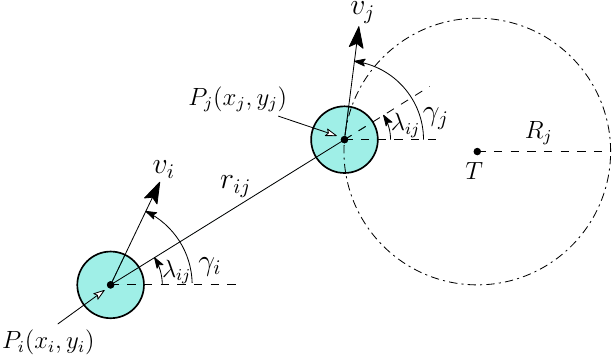}
    \caption{Engagement geometry}
    \label{fig:engage_geom}
\end{figure}
\begin{figure}[ht]
    \centering
    \includegraphics[scale=0.95]{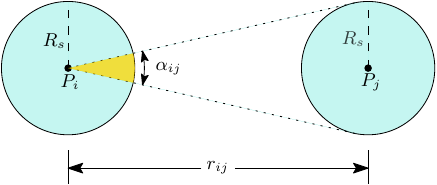}
    \caption{Sensing setup with agent $i$ observing agent $j$}
    \label{fig:sub_angle}
\end{figure}
For $i\neq j$, let $r_{ij}\coloneqq\|P_j-P_i\|$, and let $\lambda_{ij}\in\mathbb S^1$ be the line-of-sight (LOS) angle from $i$ to $j$, i.e., $P_j-P_i=r_{ij}\,\nvec{\lambda_{ij}}$ (see Fig.~\ref{fig:engage_geom}). Agent $i$ senses the agents in $\mathcal{S}_i:=\{j\neq i:\ r_{ij}\le\rho_i\}$, where $\rho_i>0$ is its sensing radius. For each $j\in\mathcal{S}_i$, it identifies agent $j$ and measures the angle $\alpha_{ij}$ subtended by agent $j$ at $P_i$ (see Fig.~\ref{fig:sub_angle}) and the relative angles
\setlength{\abovedisplayskip}{2pt}
\setlength{\belowdisplayskip}{2pt}
\begin{equation}
\label{eq:errors}
e_{i_1}^j:=\gamma_j-\gamma_i,\qquad e_{i_2}^j:=\lambda_{ij}-\gamma_i .
\end{equation}
Since every disk has radius $R_s$,
\begin{equation}\label{eq:rij}
r_{ij}=R_s/\sin(\alpha_{ij}/2),
\end{equation}
so $r_{ij}>2R_s$ if and only if $\alpha_{ij}<\pi/3$. This sensing model gives rise to the following interaction graph

\subsection{Interaction graph}
\label{subsec:inter_graph}
Let $\mathcal{V}(t)\subseteq\mathcal{V}$ be the set of agents present at time $t$. Each follower $i\in\mathcal{F}\cap\mathcal{V}(t)$ that senses at least one agent in $\mathcal{V}(t)$ takes as its out-neighbour the nearest sensed agent,
\setlength{\abovedisplayskip}{2pt}
\setlength{\belowdisplayskip}{2pt}
\begin{equation}\label{eq:nn}
j_i(t)\in\arg\max_{j\in\mathcal{S}_i(t)\cap\mathcal{V}(t)}\alpha_{ij}(t),
\end{equation}
with ties broken by a fixed rule; by eqn.~\eqref{eq:rij}, the largest $\alpha_{ij}$ belongs to the nearest agent. Leaders have no out-neighbours as they know the target location. This defines $\mathcal{G}(t)=(\mathcal{V}(t),\mathcal{E}(t))$ with $\mathcal{E}(t):=\{(i,j_i(t)):\ i\in\mathcal{F}\cap\mathcal{V}(t)\}$ (see Fig.~\ref{fig:to_communication}).
\begin{figure}[ht]
    \centering
        \begin{subfigure}[b]{0.2\textwidth}
        \centering    
        \includegraphics[width=\textwidth]{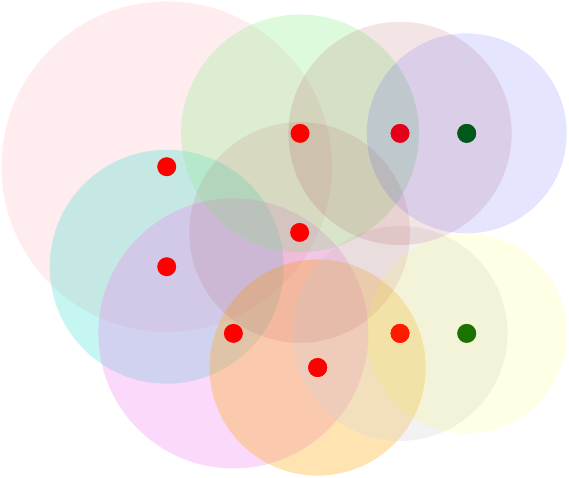}
        \caption{Distribution of agents in  $\mathbb{R}^2$}
        \label{fig:sensing_graph}
    \end{subfigure}
    \hfill
    \begin{subfigure}[b]{0.24\textwidth}
        \centering
        \includegraphics[width=0.87\textwidth]{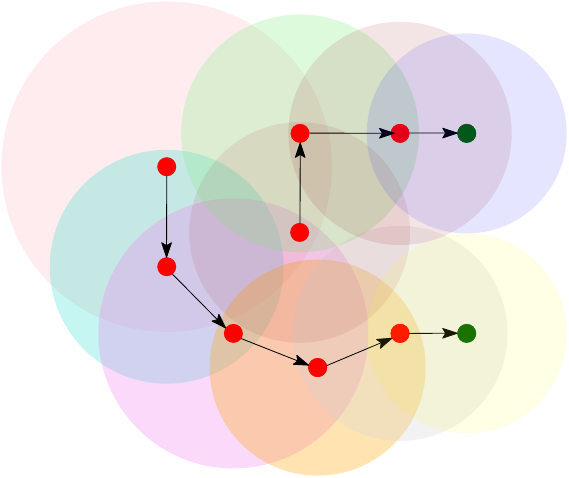}  
        \caption{Nearest neighbour based interaction topology $\mathcal{G}$}
        \label{fig:sense_to_communication_graph}
    \end{subfigure}
    \caption{Nearest neighbour interaction graph; each arrow points from a follower to its out-neighbour. Leaders are shown in green and followers in red.}
    \label{fig:to_communication}
\end{figure}
To establish the conditions necessary for information flow to all followers, we impose the following assumptions on the interaction graphs defined above.
\begin{assumption}[Path to a leader]
    \label{assump:infor_flow}
    For all $t\geqslant0$, every follower in $\mathcal{V}(t)$ senses at least one agent in $\mathcal{V}(t)$ and has a directed path (in $\mathcal{G}(t)$) to a leader.
\end{assumption}

Since each follower has exactly one out-neighbour, this path is unique; its length is the depth of follower $i$ at time $t$.
\begin{assumption}
    \label{assump:switching_graph}
    $\mathcal{G}(t)$ is right continuous and changes (which also includes agents leaving or joining) only at finitely many instants $0<t_1<t...<t_K$, $K>0$, we set $t_0=0$ for observation. Leaders are present for all $t\geqslant0$.
\end{assumption}

With these assumptions, the problem can now be formal:
\begin{problem}[Distributed circumnavigation]
  Consider the agents defined in eqn.~\eqref{eq:unicycle_kinematics}, with initial state in $\mathcal{X}_{adm}$, the sensing model as proposed in Section~\ref{subsec:sensing_model}, and the graph it forms as in Section~\ref{subsec:inter_graph} under Assumptions~\ref{assump:infor_flow} and \ref{assump:switching_graph}. Design $u_\ell=u_\ell(P_\ell,\gamma_\ell)$ for $\ell\in\mathcal{L}$ and $u_i=u_i(\alpha_{ij},e_{i_1}^j,e_{i_2}^j)$ with $j=j_i(t)$ for $i\in\mathcal{F}$, and give conditions under which:
  \begin{enumerate}
    \item[(P1)] each leader $\ell\in\mathcal{L}$ reaches $r_\ell(t)=R_\ell$ in finite time and maintains it thereafter;
    \item[(P2)] for each follower $i$ after $t_K$, $r_i(t)\to R_i$ and $\dot\gamma_i(t)\to \bar{\omega}$;
    \item[(P3)] $\|P_\ell(t)-P_m(t)\|>2R_s$ for all $t\geqslant 0$ for all distinct $\ell,m\in\mathcal{L}$, and each follower $i$ satisfies $\|P_i(t)-P_{j_i(t)}(t)\|>2R_s$.
 \end{enumerate}
\end{problem}
We now propose the guidance laws for circumnavigation.
%

\section{Guidance law for leaders}
\label{sec:guidance_for_leaders}
In our problem, only a small subset of agents are leaders, each with knowledge of the target information. This simplifies the design of guidance laws to circumnavigate while avoiding inter-agent collision. Moreover, all information in the group originates with the leaders (Assumption~\ref{assump:infor_flow}), so the followers can settle only after the leaders have. We therefore first develop the guidance law for the leaders.
\begin{figure}[ht]
\begin{center}
\includegraphics[scale=0.8]{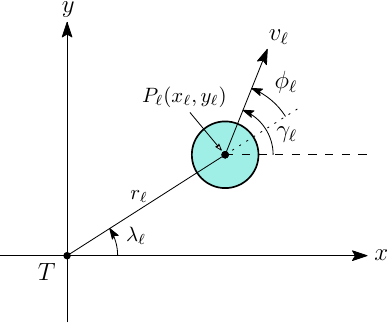}    
\caption{Engagement geometry of a leader}  
\label{figure:engagement_geometry_leader}                                 
\end{center}                                 
\end{figure}

A leader knows the target position, so it can compute its distance $r_\ell$ to the target and its lead angle $\phi_\ell$; we therefore describe its motion in polar coordinates centred at the target. In these coordinates, leader $\ell \in \mathcal{L}$ obeys (see Fig.~\ref{figure:engagement_geometry_leader})
\setlength{\abovedisplayskip}{2pt}
\setlength{\belowdisplayskip}{2pt}
\begin{equation}
\label{eq:leader_polar_kinematics}
\dot r_\ell=v_\ell\cos\phi_\ell,\quad
\dot\lambda_\ell=\frac{v_\ell}{r_\ell}\sin\phi_\ell, \quad
\dot\phi_\ell=u_\ell-\frac{v_\ell}{r_\ell}\sin\phi_\ell.
\end{equation}

The design has two stages: the lead angle is first steered to a desired value at which the radial speed equals a prescribed profile, and this profile then drives the radial error to zero, both in finite time. Define the radial error $\varepsilon_\ell \coloneqq r_\ell - R_\ell$, with $R_\ell$ from eqn.~\eqref{eq:R_k}. Choose constants $k_{r\ell}>0$, $k_{\phi \ell}>0$, $0<\mu_\ell<1$, $0<\nu_\ell<1$. The desired radial velocity is selected as $w_\ell(\varepsilon_\ell) = -k_{r\ell}|\varepsilon_\ell|^{\mu_\ell}\sgn(\varepsilon_\ell)$. Define the tangential speed $s_\ell(\varepsilon_\ell) = \sqrt{v_\ell^2-w_\ell^2(\varepsilon_\ell)}$, where $s_\ell \in \mathbb{R}_{>0}$ while $k_{r\ell}|\varepsilon_\ell|^{\mu^\ell}<v_\ell$. The desired lead angle, which splits the speed $v_\ell$ into the radial component $w_\ell$ and the tangential component $s_\ell$, is
\setlength{\abovedisplayskip}{2pt}
\setlength{\belowdisplayskip}{2pt}
\begin{equation}
\label{eq:leader_desired_heading}
\phi_\ell^\star = \arctan\!\left( s_\ell(\varepsilon_\ell), w_\ell(\varepsilon_\ell)\right).
\end{equation}
By construction,
\setlength{\abovedisplayskip}{2pt}
\setlength{\belowdisplayskip}{2pt}
\begin{equation}
\label{eq:leader_radial_tangential_identity}
v_\ell\cos\phi_\ell^\star=w_\ell(\varepsilon_\ell), \qquad v_\ell\sin\phi_\ell^\star= s_\ell(\varepsilon_\ell).
\end{equation}
Since $\phi_\ell^\star$ changes with $\varepsilon_\ell$, the guidance law must also compensate for its rate of change. For $\varepsilon_\ell\neq 0$, $w_\ell'(\varepsilon_\ell) =
\frac{d w_\ell}{d\varepsilon_\ell} = -k_{r\ell}\mu_\ell|\varepsilon_\ell|^{\mu_\ell-1}$.
 Using eqn.~\eqref{eq:leader_desired_heading}, the term $\dot{\phi}_\ell^{\star}$ is
 \setlength{\abovedisplayskip}{2pt}
\setlength{\belowdisplayskip}{2pt}
\begin{equation}
\label{eq:leader_chistar_dot}
\dot\phi_\ell^\star = - \frac{w_\ell'(\varepsilon_\ell)\dot r_\ell} {s_\ell(\varepsilon_\ell)} = - \frac{w_\ell'(\varepsilon_\ell)v_\ell\cos\phi_\ell} {s_\ell(\varepsilon_\ell)}.
\end{equation}
At $\varepsilon_\ell=0$, the desired circle is reached, and the steady circumnavigation is obtained by setting $\dot\phi_\ell^\star=0$. 

The proposed leader guidance law (if $\tilde\phi_\ell \coloneqq (\phi_\ell-\phi_\ell^\star)$) is
\setlength{\abovedisplayskip}{2pt}
\setlength{\belowdisplayskip}{2pt}
\begin{equation}
\label{eq:leader_guidance_law}
u_\ell = \frac{v_\ell}{r_\ell}\sin\phi_\ell +\dot\phi_\ell^\star -k_{\phi\ell}|\tilde\phi_\ell|^{\nu_\ell}\sgn(\tilde\phi_\ell).
\end{equation}

The first term of eqn.~\eqref{eq:leader_guidance_law} cancels the term $-(v_\ell/r_\ell)\sin\phi_\ell$ in eqn.~\eqref{eq:leader_polar_kinematics}, which changes $\phi_\ell$ even when $u_\ell=0$ because $\lambda_\ell$ rotates as the leader moves; the second term feeds forward $\dot\phi_\ell^\star$, and the third drives $\tilde\phi_\ell$ to zero in finite time. Both convergence times and the range of $r_\ell$ during the transient can be bounded from the initial conditions. If $k_{r\ell}|\varepsilon_\ell(0)|^{\mu_\ell}<v_\ell$, then $\phi^\star_\ell(0)$ is defined; let
\setlength{\abovedisplayskip}{2pt}
\setlength{\belowdisplayskip}{2pt}
\begin{equation}
\label{eq:Tl}
\begin{aligned}
T_{\phi,\ell}&:=\frac{|\tilde\phi_\ell(0)|^{1-\nu_\ell}}{k_{\phi\ell}(1-\nu_\ell)},\qquad \bar\varepsilon_\ell:=|\varepsilon_\ell(0)|+v_\ell T_{\phi,\ell},\\ T_\ell&:=T_{\phi,\ell}+\frac{\bar\varepsilon_\ell^{\,1-\mu_\ell}}{k_{r\ell}(1-\mu_\ell)},
\end{aligned}
\end{equation}
and $I_\ell:=\big[\min\{R_\ell,\,r_\ell(0)-v_\ell T_{\phi,\ell}\},\ \max\{R_\ell,\,r_\ell(0)+v_\ell T_{\phi,\ell}\}\big]$. The following result shows that each leader reaches its circle by the time $T_\ell$, and gives a condition, checkable at $t=0$, under which no two leaders collide.
\begin{thm}
\label{thm:leader_guidance}
    Consider any leader $\ell\in\mathcal{L}$ that applies control input given in eqn.~\eqref{eq:leader_guidance_law}, with $r_\ell(0)>v_\ell T_{\phi,\ell}$ and $k_{r\ell}\bar\varepsilon_\ell^{\,\mu_\ell}<v_\ell$. Then $c_\ell(t)=0$ for all $t\geqslant T_\ell$ and all $\ell\in\mathcal{L}$. If, in addition, $\dist(I_\ell,I_m)>2R_s$ for all distinct $\ell,m\in\mathcal{L}$, then $\|P_\ell(t)-P_m(t)\|>2R_s$ for all $t\geqslant 0$.
\end{thm}
\begin{proof}
From eqns.~\eqref{eq:leader_polar_kinematics} and \eqref{eq:leader_guidance_law}, $\dot\phi_\ell = u_\ell-{(v_\ell)}/{(r_\ell)}\sin\phi_\ell = \dot\phi_\ell^\star -k_{\phi\ell}|\tilde\phi_\ell|^{\nu_\ell}\sgn(\tilde\phi_\ell)$. Hence, $\dot{\tilde\phi}_\ell = -k_{\phi\ell}|\tilde\phi_\ell|^{\nu_\ell}\sgn(\tilde\phi_\ell).$

With Lyapunov candidate function $V_{\phi,\ell}=(1/2)\tilde\phi_\ell^2$, $\dot V_{\phi,\ell} = -k_{\phi\ell}|\tilde\phi_\ell|^{\nu_\ell+1}<0$ for $\tilde\phi_\ell\neq0$. Moreover, $({d}/{dt})z_\ell^{1-\nu_\ell} = -k_{\phi\ell}(1-\nu_\ell)$ while $\tilde{\phi}_\ell \neq 0$, which gives $|\tilde\phi_\ell(t)|^{1-\nu_\ell} = |\tilde\phi_\ell(0)|^{1-\nu_\ell} -k_{\phi\ell}(1-\nu_\ell)t$. Thus $\tilde\phi_\ell$ reaches zero in finite time $T_{\phi,\ell}$. On $[0,T_{\phi,\ell}]$, $|\dot{r}_\ell| \leqslant v_\ell$; hence $r_\ell(t) \geqslant r_\ell (0)-v_\ell T_{\phi,\ell}>0$, so eqn.~\eqref{eq:leader_polar_kinematics} is well defined, and $|\varepsilon_\ell(t)| \leqslant \bar{\varepsilon}_\ell$, so $s_\ell(\varepsilon_\ell(t))>0$ because $k_{r\ell}\bar{\varepsilon}_\ell^{\mu_\ell}<v_\ell$.

For $t\geqslant T_{\phi,\ell}$, $\phi_\ell=\phi_\ell^\star$. Hence, using eqn.~\eqref{eq:leader_radial_tangential_identity}, $\dot\varepsilon_\ell = \dot r_\ell = v_\ell\cos\phi_\ell^\star = w_\ell(\varepsilon_\ell) = -k_{r\ell}|\varepsilon_\ell|^{\mu_\ell}\sgn(\varepsilon_\ell).$ Hence $|\varepsilon_\ell(t)|^{1-\mu_\ell} = |\varepsilon_\ell(T_{\phi,\ell})|^{1-\mu_\ell} -k_{r\ell}(1-\mu_\ell)(t-T_{\phi,\ell})$ until $\varepsilon_\ell=0$. Since $0<\mu_\ell<1$, $\varepsilon_\ell$ reaches zero in finite time, at the latest at $T_\ell$ because $|\varepsilon_\ell(T_{\phi,\ell})|\leqslant \bar{\varepsilon}_\ell$. Therefore $r_\ell=R_\ell$ in finite time. At $\varepsilon_\ell=0$, $w_\ell(0)=0$ and $s_\ell(0)=v_\ell$, so $\phi_\ell^\star=\atantwo( v_\ell,0).$
Thus, $\cos\phi_\ell^\star=0$ and $\sin\phi_\ell^\star=1$.

Hence, for $t\geqslant T_\ell$, $r_\ell=R_\ell$ and $\phi_\ell=\pi/2$, i.e., $P_\ell=R_\ell \nvec{\lambda_\ell}$ and $\gamma_\ell=\lambda_\ell+\pi/2$, so eqn.~\eqref{eq:turn_centre} gives $c_\ell=R_\ell \nvec{\lambda_\ell}+R_\ell \nvec{\lambda_\ell+\pi}=0$. Now, it remains to show the stated leader-leader safety condition. Before $T_{\phi,\ell}$, $|\dot r_\ell|=|v_\ell\cos\phi_\ell|\leqslant v_\ell$ so $r_\ell(t)\in \left[ r_\ell(0)-v_\ell T_{\phi,\ell}, \, r_\ell(0)+v_\ell T_{\phi,\ell} \right].$

After $T_{\phi,\ell}$, $\dot\varepsilon_\ell$ has the opposite sign of $\varepsilon_\ell$, so $r_\ell(t)$ moves monotonically toward $R_\ell$. Therefore $r_\ell(t)\in\mathcal I_\ell$ $\forall$ $t\geqslant 0$.

For any two leaders $\ell,m$, the reverse triangle inequality gives $\|P_\ell-P_m\|
\geqslant
\big|\|P_\ell\|-\|P_m\|\big|
=
|r_\ell-r_m|.$
Since $r_\ell(t)\in\mathcal I_\ell$ and
$r_m(t)\in\mathcal I_m$, $|r_\ell(t)-r_m(t)|
\geqslant
\dist(\mathcal I_\ell,\mathcal I_m)>2R_s.$
Hence $\|P_\ell(t)-P_m(t)\|>2R_s$ for all $t\geqslant 0$.
\end{proof}

Theorem~\ref{thm:leader_guidance} does not bound the turn rate: by eqn.~\eqref{eq:leader_chistar_dot}, $\dot\phi_\ell^\star$ is proportional to $|\varepsilon_\ell|^{\mu_\ell-1}\cos\phi_\ell$, which is unbounded as $\varepsilon_\ell\to0$ unless $\cos\phi_\ell$ vanishes fast enough. Since the follower analysis requires bounded leader inputs, we give conditions under which $u_\ell$ stays bounded.
\begin{corollary}[Bounded leader input]\label{cor:leader-input}
Under the hypotheses of Theorem~\ref{thm:leader_guidance}, let $\mu_\ell\geqslant 1/2$. Then $u_\ell$ is bounded on $[T_{\phi,\ell},\infty)$. If, in addition, $|\varepsilon_\ell(0)|>v_\ell T_{\phi,\ell}$, then $u_\ell$ is bounded on $[0,\infty)$.
\end{corollary}
\begin{proof}
By the proof of Theorem~\ref{thm:leader_guidance}, $r_\ell(t)\geqslant \min I_\ell>0$ and $|\varepsilon_\ell(t)|\leqslant \bar\varepsilon_\ell$ for all $t\geqslant 0$, so $s_\ell(\varepsilon_\ell(t))\geqslant \underline s_\ell:=(v_\ell^2-k_{r\ell}^2\bar\varepsilon_\ell^{\,2\mu_\ell})^{1/2}>0$, and the first and last terms of eqn.~\eqref{eq:leader_guidance_law} are bounded. For $t\geqslant T_{\phi,\ell}$, $v_\ell\cos\phi_\ell=w_\ell(\varepsilon_\ell)$ by eqn.~\eqref{eq:leader_radial_tangential_identity}, so eqn.~\eqref{eq:leader_chistar_dot} gives $|\dot\phi^\star_\ell|=\mu_\ell k_{r\ell}^2|\varepsilon_\ell|^{2\mu_\ell-1}/s_\ell(\varepsilon_\ell)\leqslant \mu_\ell k_{r\ell}^2\bar\varepsilon_\ell^{\,2\mu_\ell-1}/\underline s_\ell$. For $t\in[0,T_{\phi,\ell}]$, $|\varepsilon_\ell(t)|\geqslant |\varepsilon_\ell(0)|-v_\ell T_{\phi,\ell}>0$, and eqn.~\eqref{eq:leader_chistar_dot} gives $|\dot\phi^\star_\ell|\leqslant \mu_\ell k_{r\ell}v_\ell\big(|\varepsilon_\ell(0)|-v_\ell T_{\phi,\ell}\big)^{\mu_\ell-1}/\underline s_\ell$.
\end{proof}

After $T_\ell$, each leader moves on its circle with $c_\ell=0$ and $u_\ell=\bar\omega$, and thus provides a fixed reference for the followers, who do not know the target position.

\section{Guidance law for followers}
\label{sec:follower_guidance}
A follower knows neither the target position nor its distance to the target; it measures only three angles relative to its out-neighbour. We proceed in four steps: we show that these angles determine the follower's turn centre once its out-neighbour has converged; characterise the collision-free steady-state formations; design the guidance law; and analyse it.

\subsection{Turn centre}
\label{subsec:turn_centre}
By eqn.~\eqref{eq:turn_centre}, $c_k$ is the centre of the circle that agent $k$ traces at the nominal rate $\bar\omega$, so agent $k$ circles the target at this rate when $c_k$ lies at the origin. The follower law is built on this observation. However, $c_i$ depends on the follower's position, which the follower does not know; we now show how to compute it from the measurements. For a follower $i$ with out-neighbour $j$, define the measured quantity
\setlength{\abovedisplayskip}{2pt}
\setlength{\belowdisplayskip}{2pt}
\begin{equation}
\label{eq:Lambda}
\Lambda_i \coloneqq R_j\sin e_{i_1}^j-r_{ij}\cos e_{i_2}^j,
\end{equation}
where $r_{ij}$ is computed from $\alpha_{ij}$ by eqn.~\eqref{eq:rij}. The following lemma collects the properties of the turn centre used.
\begin{lemma}[Turn centre]
\label{lem:tc}
Along eqn.~\eqref{eq:unicycle_kinematics}, for every $k\in\mathcal{V}$ and every follower $i$ with out-neighbour $j$:
\begin{enumerate}
\item[(a)] $\dot c_k=-R_k(u_k-\bar{\omega})\,\nvec{\gamma_k$};
\item[(b)] agent $k$ moves on the circle of radius $R_k$ centred at the origin with angular rate $\bar{\omega}$ if and only if $c_k\equiv0$;
\item[(c)] $\Rot{-\gamma_i}(c_i-c_j)=\big[\Lambda_i\ \ \ R_i-R_j\cos e_{i_1}^j-r_{ij}\sin e_{i_2}^j\big]^\top$; hence $|\langle c_i,\nvec{\gamma_i}\rangle-\Lambda_i|\le\|c_j\|$, and $\langle c_i,\nvec{\gamma_i}\rangle=\Lambda_i$ if $c_j=0$.
\end{enumerate}
\end{lemma}
\begin{proof}
 Writing eqn.~\eqref{eq:turn_centre} componentwise, $c_k=[\,x_k-R_k\sin\gamma_k\ \ \ y_k+R_k\cos\gamma_k\,]^T$, and differentiating each row along eqn.~\eqref{eq:unicycle_kinematics}, $\dot c_{k,x}=v_k\cos\gamma_k-R_ku_k\cos\gamma_k=(v_k-R_ku_k)\cos\gamma_k,$ and $\dot c_{k,y}=v_k\sin\gamma_k-R_ku_k\sin\gamma_k=(v_k-R_ku_k)\sin\gamma_k.$ So that $\dot c_k=(v_k-R_ku_k)\nvec{\gamma_k}$. Substituting $v_k=R_k\bar\omega$ from eqn.~\eqref{eq:R_k} gives $v_k-R_ku_k=-R_k(u_k-\bar\omega)$, which is condition (a).

 If $c_k\equiv 0$, then eqn.~\eqref{eq:turn_centre} gives $P_k=-R_k\nvec{\gamma_k+\pi/2}$, so $\|P_k\|=R_k$ and the heading is tangent to that circle; moreover $\dot c_k=0$ forces $u_k=\bar\omega$ by (a). Conversely, on that circle $\gamma_k=\lambda_k+\pi/2$ and $P_k-T=R_k\nvec{\lambda_k}$, so eqn.~\eqref{eq:turn_centre} returns $c_k=0$.

For condition (c), eqn.~\eqref{eq:turn_centre} gives $P_j=c_j+R_j\nvec{\gamma_j-\pi/2}$ and $P_i=P_j-r_{ij}\nvec{\lambda_{ij}}$. Substituting into eqn.~\eqref{eq:turn_centre} for $k=i$,
$c_i-c_j=R_j\nvec{\gamma_j-\pi/2}-r_{ij}\nvec{\lambda_{ij}}+R_i\nvec{\gamma_i+\pi/2}$.
Multiplying by $\Rot{-\gamma_i}$ and using eqn.~\eqref{eq:turn_centre}, $\Rot{-\gamma_i}(c_i-c_j)=R_j\nvec{e_{i_1}^j-\pi/2}-r_{ij}\nvec{e_{i_2}^j}+R_i\nvec{\pi/2}$, whose components are $\Lambda_i$, by eqn.~\eqref{eq:Lambda}, and $R_i-R_j \cos e_{i_1}^j-r_{ij}\sin e_{i_2}^j$. Finally, $\langle c_i,\nvec{\gamma_i}\rangle-\Lambda_i=\langle c_j,\nvec{\gamma_i}\rangle$ and $|\langle c_j,\nvec{\gamma_i} \rangle|\leq \|c_j\|$.
\end{proof}

By Lemma~\ref{lem:tc}, follower $i$ evaluates $\langle c_i,\nvec{\gamma_i} \rangle$ from $(\alpha_{ij},e_{i_1}^j,e_{i_2}^j)$ once $c_j=0$. In fact, once $c_j=0$, part (c) gives both components of $\Rot{-\gamma_i}c_i$ from measured quantities; the guidance law needs only the first, $\Lambda_i$. By part (a), the turn centre moves only when $u_i\neq\bar\omega$, and its motion is set directly by $u_i$, which makes $c_i$ a convenient state for the design.

\subsection{Steady-state formation}
\label{subsec:steady_state}
Before designing the law, we identify the steady states it should produce: the values of $e_{i_1}^j$ and $e_{i_2}^j$ at which follower $i$ circles the target with its out-neighbour, and among them the collision-free ones. The following lemma shows that keeping both errors constant suffices.
\begin{lemma}[Common centre]
\label{lem:cc}
Let follower $i$ have out-neighbour $j$ with $c_j\equiv0$. If $e_{i_1}^j$ and $e_{i_2}^j$ are constant, then $c_i\equiv0$.
\end{lemma}
\begin{proof}
By Lemma~\ref{lem:tc}, $u_j=\bar{\omega}$; since $e_{i_1}^j$ is constant, $u_i=u_j=\bar{\omega}$, so $\dot c_i=0$ by Lemma~\ref{lem:tc}. By eqn.~\eqref{eq:turn_centre} and $c_j=0$, $r_{ij}\,\nvec{e_{i_2}^j}=\Rot{-\gamma_i}(P_j-P_i)=b_{ij}-\Rot{-\gamma_i}c_i$, where $b_{ij}:=R_j\nvec{e_{i_1}^j-\pi/2}-R_i\nvec{-\pi/2}$ is constant. Since $\dot\gamma_i=\bar{\omega}\neq0$, the right-hand side traverses the circle of radius $\|c_i\|$ centred at $b_{ij}$, so its direction $e_{i_2}^j$ is constant only if $c_i=0$. Hence, proved.
\end{proof}

By Lemma~\ref{lem:cc}, holding $e_{i_1}^j$ and $e_{i_2}^j$ constant while $c_j \equiv 0$ makes follower $i$ circumnavigate the target. The two constants cannot be chosen independently: once $e_{i_1}^j$ is fixed, the geometry of the two circles fixes $e_{i_2}^j$ and the inter-agent distance. We therefore use the steady-state value of $e_{i_1}^j$ as the only design parameter of each edge.
For a follower $i$ with out-neighbour $j$, the formation parameter $\bar{e}_{i_1}^j\in(-\pi,\pi]$ is the desired steady-state value of $e_{i_1}^j$. Let
\setlength{\abovedisplayskip}{2pt}
\setlength{\belowdisplayskip}{2pt}
\begin{align}
\bar r_{ij}& \coloneqq \sqrt{R_i^2+R_j^2-2R_iR_j\cos{\bar{e}_{i_1}^j}},\label{eq:rbar}\\
\bar{e}_{i_2}^j&:=\atantwo{\big(R_i-R_j\cos{\bar{e}_{i_1}^j},\ R_j\sin{\bar{e}_{i_2}^j}\big)}.\label{eq:e2bar}
\end{align}
By eqn.~\eqref{eq:rbar}, $\bar r_{ij}>2R_s$ if and only if
\setlength{\abovedisplayskip}{2pt}
\setlength{\belowdisplayskip}{2pt}
\begin{equation}
\label{eq:coscond}
\cos{\bar{e}_{i_1}^j}<\frac{R_i^2+R_j^2-4R_s^2}{2R_iR_j}.
\end{equation}
Such $\bar{e}_{i_1}^j$ exists if and only if $R_i+R_j>2R_s$. 
The next lemma confirms that $\bar r_{ij}$ and $\bar e_{i_2}^j$ are the distance and line-of-sight error that this choice produces, so that the condition in eqn.~\eqref{eq:coscond} makes the steady state collision-free.
\begin{lemma}[Steady-state formation]
\label{lem:ss}
If $c_i\equiv c_j\equiv0$ and $e_{i_1}^j \equiv \bar{e}_{i_1}^j$, then $r_{ij}\equiv\bar r_{ij}$ and $e_{i_2}^j \equiv \bar{e}_{i_2}^j$.
\end{lemma}
\begin{proof}
By Lemma~\ref{lem:tc}, $r_k=R_k$ and $\gamma_k=\lambda_k+\pi/2$ for $k\in \{i,j\}$, hence $e_{i_1}^j=\gamma_j-\gamma_i=\lambda_j-\lambda_i$ is exactly the phase separation between the two agents. Applying the law of cosines to the triangle with vertices $T$, $P_i$, $P_j$, whose sides adjacent to $T$ have lengths $R_i$ and $R_j$ and include the angle $\bar e_{i_1}^j$, gives $r_{ij}=\bar{r}_{ij}$. For $e_{i_2}^j$, express $P_j-P_i$ in the body frame of agent $i$:
\setlength{\abovedisplayskip}{2pt}
\setlength{\belowdisplayskip}{2pt}
\begin{equation}
\label{eq:body_frame_rij}
\begin{aligned}
\Rot{-\gamma_i}(P_j-P_i)
&=R_j\nvec{\bar e_{i_1}^j-\tfrac{\pi}{2}}
  -R_i\nvec{-\tfrac{\pi}{2}}\\
&=\begin{bmatrix}
R_j\sin\bar e_{i_1}^j\\
R_i-R_j\cos\bar e_{i_1}^j
\end{bmatrix},
\end{aligned}
\end{equation}
and $e_{i_2}^j=\lambda_{ij}-\gamma_i$ is the argument of this vector, i.e., $e_{i_2}^j=\bar{e}_{i_2}^j$ by eqn.~\eqref{eq:e2bar}. Hence, proved.
\end{proof}

We choose $\bar{e}_{i_1}^j$ to satisfy eqn.~\eqref{eq:coscond}; then $\bar e_{i_2}^j$ and $\bar r_{ij}$ follow from eqns.~\eqref{eq:rbar}--\eqref{eq:e2bar}. Since the relative phases of all agents are obtained by composing the parameters $\bar e_{i_1}^j$ along the directed paths of $\mathcal G$, the steady state formation of the agents connected to a common leader is known at design time, up to that leader's phase. It remains to design a guidance law that drives each follower to this steady state using only its three measurements.

\subsection{Guidance law design}
\label{subsec:design}
By Lemma~\ref{lem:tc}(c), the follower can compute its turn centre only after its out-neighbour has settled, and the steady states of Section~\ref{subsec:steady_state} also presume $c_j\equiv0$. We therefore design the law for a follower whose out-neighbour has converged, as stated next, and treat the remaining cases afterwards.
\begin{assumption}[Converged out-neighbour]
\label{ass:conv}
For follower $i$, there is an agent $j$ and a finite $T_j\geqslant 0$ such that $j_i(t)=j$ and $c_j(t)=0$ for all $t\geqslant T_j$.
\end{assumption}

By Theorem~\ref{thm:leader_guidance}, Assumption~\ref{ass:conv} holds with $T_j=T_\ell$ whenever $j_i(t)\equiv\ell\in\mathcal{L}$; by Lemma~\ref{lem:tc}, it implies $\dot\gamma_j(t)=\bar{\omega}$ $\forall$ $t\geqslant T_j$.

We measure the deviation from the steady state of Lemma~\ref{lem:ss} by two angular errors. Let $\bar{e}_{i_1}^j$ satisfy eqn.~\eqref{eq:coscond}, with $\bar{e}_{i_2}^j$ and $\bar r_{ij}$ given by eqns.~\eqref{eq:rbar}-\eqref{eq:e2bar}. Define
\setlength{\abovedisplayskip}{2pt}
\setlength{\belowdisplayskip}{2pt}
\begin{equation}
\label{eq:zdef}
z_i\coloneqq e_{i_1}^j-\bar{e}_{i_1}^j,\qquad \zeta_i:=e_{i_2}^j-\bar{e}_{i_2}^j .
\end{equation}
A barrier term keeps the relative heading error $z_i$ inside a prescribed interval. This excludes the undesired equilibrium at which the follower settles half a turn away from its desired phase and, together with the Lyapunov function below, bounds the inter-agent distance error.

\emph{Barrier and gains:}  Choose a barrier of width $\eta_i\in(0,\pi)$ and gains $\kappa_{i1},\kappa_{i2}>0$, and define the logarithmic BLF term and its derivative,
\setlength{\abovedisplayskip}{2pt}
\setlength{\belowdisplayskip}{2pt}
\begin{equation}
\label{eq:Pi}
\Pi_i(z_i)\coloneqq -\kappa_{i2}\ln\!\left(1-\frac{z_i^2}{\eta_i^2}\right),
\qquad
\Pi_i'(z_i)=\frac{2\kappa_{i2}z_i}{\eta_i^2-z_i^2}.
\end{equation}
The barrier constrains only $z_i$. Convergence also needs an error that measures how far the follower is from its circle; the turn centre provides it, and Lemma~\ref{lem:tc}(c) makes it available from the measurements.

\emph{Error coordinates:} Define the body-frame centre error
\setlength{\abovedisplayskip}{2pt}
\setlength{\belowdisplayskip}{2pt}
\begin{equation}
\label{eq:chi}
\chi_i\triangleq\Rot{-\gamma_i}c_i=[\chi_{i,x}\ \ \chi_{i,y}]^T\in\mathbb R^2,
\end{equation}
so that $\|\chi_i\|=\|c_i\|$ and, by Lemma~\ref{lem:tc}, for $t\geqslant T_j$,
\setlength{\abovedisplayskip}{2pt}
\setlength{\belowdisplayskip}{2pt}
\begin{equation}
\label{eq:chimeas}
\chi_{i,x}=\Lambda_i,\qquad \chi_{i,y}=R_i-R_j\cos{e_{i_1}^j}-r_{ij}\sin{e_{i_2}^j}.
\end{equation}
Writing $\tilde{u}_i \coloneqq u_i -\bar{\omega}$ and using $\frac{d}{dt}\Rot{-\gamma_i}=-u_i J \Rot{-\gamma_i}$ together with Lemma~\ref{lem:tc} and $\Rot{-\gamma_i}\nvec{\gamma_i}=[1 \quad 0]^T$, we obtain $\dot\chi_i=-u_iJ\chi_i-R_i\tilde u_i[1\ \ 0]^T$, that is, for any input $u_i$ and $t\geqslant T_j$,
\begin{equation}
\label{eq:closed}
\dot\chi_{i,x}=u_i\chi_{i,y}-R_i\tilde u_i,
\quad
\dot\chi_{i,y}=-u_i\chi_{i,x},
\quad
\dot z_i=-\tilde u_i,
\end{equation}
where the last identity uses $\dot{\gamma}_j=\bar{\omega}$. Let $\mathcal{D}_i \coloneqq \mathbb{R}^2 \times (-\eta_i, \eta_i)$ be the constraint set.
The angular errors $(z_i,\zeta_i)$ of eqn.~\eqref{eq:zdef} alone can be expected to describe the error dynamics; the following remark explains why the turn centre is needed.
\begin{remark}[State dimension]
\label{rem:state}
The pair $(z_i,\zeta_i)$ is not a state of the follower error dynamics: for $t\geqslant T_j$, $\dot z_i-\dot\zeta_i=\bar{\omega}-\dot\lambda_{ij}$, and $\dot\lambda_{ij}$ depends on $r_{ij}$. The coordinate $\chi_{i,y}$ in eqn.~\eqref{eq:closed} carries this distance information.
\end{remark}

The Lyapunov function adds to the barrier the squared distance of the turn centre from the target, which vanishes exactly when the follower moves along its desired circle.

\emph{Lyapunov function:} On $\mathcal{D}_i$ define
\setlength{\abovedisplayskip}{2pt}
\setlength{\belowdisplayskip}{2pt}
\begin{equation}
\label{eq:V}
V_i(\chi_i,z_i):=\frac{\kappa_{i1}}{2R_i}\|\chi_i\|^2+\Pi_i(z_i).
\end{equation}
Along eqn.~\eqref{eq:closed}, for any input $u_i$,
\setlength{\abovedisplayskip}{2pt}
\setlength{\belowdisplayskip}{2pt}
\begin{equation}
\label{eq:Vdot}
\dot V_i=-\tilde{u}_i\big(\kappa_{i1}\chi_{i,x}+\Pi_i'(z_i)\big),\qquad t\geqslant T_j .
\end{equation}
Since $\dot V_i$ is linear in $\tilde u_i$, the input can be chosen to make $\dot V_i$ negative semidefinite. Choosing $\tilde{u}_i=\kappa_{i1}\chi_{i,x}+\Pi_i'(z_i)$, and replacing $\chi_{i,x}$ by the measured $\Lambda_i$, which equals it for $t\geqslant T_j$ by eqn.~\eqref{eq:chimeas}, gives the follower guidance law, for all $t\geqslant 0$:
\begin{equation}
\label{eq:ui}
u_i=\bar{\omega}+\kappa_{i1}\Lambda_i+\frac{2\kappa_{i2}z_i}{\eta_i^2-z_i^2}.
\end{equation}
This law uses only $(\alpha_{ij},e_{i_1}^j,e_{i_2}^j)$ and design constants; it requires no range-rate measurement, no differentiation of measured signals and no state memory, is smooth on $\mathcal{D}_i$, and yields $\dot V_i=-\tilde{u}_i^2$ for $t\geqslant T_j$. Moreover, $V_i$ is positive definite on $\mathcal{D}_i$, $V_i\to\infty$ as $z_i\to\pm\eta_i$, and $V_i$ satisfies the hypotheses of Lemma~\ref{Lemma:BLF} with $z_1=z_i$, $w=\chi_i$, $k_{a_1}=k_{b_1}=\eta_i$, $V_1=\Pi_i$ and $U(w)=\kappa_{i1}\|w\|^2/(2R_i)$.

Convergence alone does not exclude a collision during the transient. Since $V_i$ does not increase after $T_j$, it suffices that the sublevel set of $V_i$ that contains the state at $T_j$ lies in the safe region; the function $\Gamma_i$ bounds the distance error $|r_{ij}-\bar r_{ij}|$ on each sublevel set. For $q\geqslant0$, let
\setlength{\abovedisplayskip}{2pt}
\setlength{\belowdisplayskip}{2pt}
\begin{equation}
\label{eq:Gamma}
\Gamma_i(q):=\sqrt{2R_iq/\kappa_{i1}}+R_j\eta_i\sqrt{1-e^{-q/\kappa_{i2}}} .
\end{equation}
The function $\Gamma_i$ is continuous and strictly increasing, with $\Gamma_i(0)=0$ and $\Gamma_i(q)\to\infty$ as $q\to\infty$. Since $\bar r_{ij}>2R_s$ by eqn.~\eqref{eq:coscond}, there is a unique $\bar q_i>0$ such that $\Gamma_i(\bar q_i)=\bar r_{ij}-2R_s$.

The guidance law defined in eqn.~\eqref{eq:ui} is applied from $t=0$, whereas the analysis above holds only after $T_j$. Before $T_j$, $\Lambda_i$ differs from $\chi_{i,x}$ and $V_i$ may increase, so we must show separately that $z_i$ does not reach the barrier.
\begin{lemma}[Barrier invariance]
\label{lem:barrier}
Let follower $i$ applies guidance law defined in eqn.~\eqref{eq:ui} on a finite interval $[a,b]$ on which $j_i(t)=j$ and $u_j$ is bounded, and let $|z_i(a)|<\eta_i$. Then the solution of follower $i$ exists on $[a,b]$, and $z_i(t)$ remains in a compact subset of $(-\eta_i,\eta_i)$; i.e., $u_i$ is bounded on $[a,b]$.
\end{lemma}
\begin{proof}
On $[a,b]$, $\dot z_i=u_j-u_i=(u_j-\bar{\omega})-\kappa_{i1}\Lambda_i-\Pi_i'(z_i)$ by eqn.~\eqref{eq:ui}. By eqn.~\eqref{eq:Lambda}, $|\Lambda_i|\le R_j+r_{ij}$, and $r_{ij}(t)\le r_{ij}(a)+(v_i+v_j)(t-a)$; hence $\bar M:=\sup_{[a,b]}|u_j-\bar{\omega}-\kappa_{i1}\Lambda_i|<\infty$. Since $\Pi_i'$ is odd and increases from $-\infty$ to $\infty$ on $(-\eta_i,\eta_i)$, there is $\bar z\in[0,\eta_i)$ with $\Pi_i'(\bar z)=\bar M$. Then $\dot z_i<0$ whenever $z_i>\bar z$, and $\dot z_i>0$ whenever $z_i<-\bar z$, so $|z_i(t)|\le\max\{|z_i(a)|,\bar z\}<\eta_i$ on the maximal interval of existence in $[a,b]$. There, $u_i$ is bounded and $P_i$ grows at most linearly, so the solution extends to all of $[a,b]$.
\end{proof}

\subsection{Convergence and safety analysis}
\label{subsec:analysis}
We now state the main result for a follower whose out-neighbour has converged: from any state in $\mathcal{D}_i$ at time $T_j$, it converges to the steady state of Lemma~\ref{lem:ss}, and it stays separated from its out-neighbour if that state lies in a small enough sublevel set of $V_i$.
\begin{thm}[Follower with a converged out-neighbour]
\label{thm:follower}
Consider an agent $i$ that applies control input defined in eqn.~\eqref{eq:ui}, let Assumption~\ref{ass:conv} hold and $|z_i(T_j)|<\eta_i$. Then:
\begin{enumerate}
\item[(a)] $|z_i(t)|<\eta_i$ for all $t$;
\item[(b)] the origin of eqn.~\eqref{eq:closed} under eqn.~\eqref{eq:ui} is asymptotically stable with region of attraction $\mathcal{D}_i$; hence $r_i(t)\to R_i$, $r_{ij}(t)\to\bar r_{ij}$ and $\dot\gamma_i(t)\to \bar{\omega}$;
\item[(c)] if $V_i(T_j)<\bar q_i$, then $r_{ij}(t)>2R_s$ for all $t$.
\end{enumerate}
\end{thm}
\begin{proof}
See Appendix~\ref{app:thm2}.
\end{proof}

By Lemma~\ref{lem:barrier} with $[a,b]=[0,T_j]$, the hypothesis $|z_i(T_j)|<\eta_i$ of Theorem~\ref{thm:follower} holds whenever $|z_i(0)|<\eta_i$, $j_i(t)\equiv j$, and $u_j$ is bounded on $[0,T_j]$; for $j\in\mathcal{L}$, Corollary~\ref{cor:leader-input} gives conditions for the last property. Through Assumption~\ref{ass:conv}, Theorem~\ref{thm:follower} presumes that $j$ remains the nearest sensed agent of $i$; choosing $\bar r_{ij}<\rho_i$ keeps the steady state within sensing range.

Theorem~\ref{thm:follower} gives asymptotic convergence. Near the steady state, the convergence is in fact exponential for all positive gains, and is used below to extend the result along the graph. Let $\kappa'_{i2}:=\Pi_i''(0)=2\kappa_{i2}/\eta_i^2$ and
\setlength{\abovedisplayskip}{2pt}
\setlength{\belowdisplayskip}{2pt}
\begin{equation}\label{eq:Ai}
A_i:=\begin{bmatrix}-R_i\kappa_{i1}&\bar{\omega}&-R_i\kappa'_{i2}\\-\bar{\omega}&0&0\\-\kappa_{i1}&0&-\kappa'_{i2}\end{bmatrix}.
\end{equation}
\begin{corollary}[Local exponential stability]
\label{cor:les}
Under the hypotheses of Theorem~\ref{thm:follower}, the origin of eqn.~\eqref{eq:closed} under eqn.~\eqref{eq:ui} is locally exponentially stable for all $\kappa_{i1},\kappa_{i2}>0$ and $\eta_i\in(0,\pi)$.
\end{corollary}
\begin{proof}
Let $X \coloneqq [\chi_{i,x}\ \ \chi_{i,y}\ \ z_i]^\top$. At the origin, $u_i=\bar{\omega}$, $\tilde{u}_i\approx\kappa_{i1}\chi_{i,x}+\kappa'_{i2}z_i$, and the products $\tilde{u}_i\chi_{i,\bullet}$ are of second order, so the linearization of eqn.~\eqref{eq:closed} is $\dot X=A_iX$ with $A_i$ in eqn.~\eqref{eq:Ai}. Expanding along the second row, $\det(sI-A_i)=\bar{\omega}^2(s+\kappa'_{i2})+s\big[(s+R_i\kappa_{i1})(s+\kappa'_{i2})-R_i\kappa_{i1}\kappa'_{i2}\big]=s^3+(R_i\kappa_{i1}+\kappa'_{i2})s^2+\bar{\omega}^2s+\bar{\omega}^2\kappa'_{i2}$. All coefficients are positive, and $(R_i\kappa_{i1}+\kappa'_{i2})\bar{\omega}^2-\bar{\omega}^2\kappa'_{i2}=\bar{\omega}^2R_i\kappa_{i1}>0$. Hence $A_i$ is Hurwitz by the Routh-Hurwitz criterion, and the claim follows from Lyapunov's indirect method.
\end{proof}
The barrier term in eqn.~\eqref{eq:ui} grows as $|z_i|\to\eta_i$, so one may ask whether the commanded turn rate stays bounded, as actuators require. It does, with a bound set by $V_i(T_j)$. For $q\ge0$, let
\setlength{\abovedisplayskip}{2pt}
\setlength{\belowdisplayskip}{2pt}
\begin{equation}
\label{eq:Bi}
B_i(q)\coloneqq \sqrt{2R_i\kappa_{i1}q}+\frac{2\kappa_{i2}}{\eta_i}\,e^{q/\kappa_{i2}}\sqrt{1-e^{-q/\kappa_{i2}}} .
\end{equation}
The barrier is what makes the region of attraction the whole of $\mathcal{D}_i$; the following remark shows what is lost with a bounded term in its place.
\begin{corollary}[Bounded input]\label{cor:input}
Under the hypotheses of Theorem~\ref{thm:follower}, $|u_i(t)-\bar{\omega}|\le B_i(V_i(T_j))$ for all $t\ge T_j$, with $B_i$ as in eqn.~\eqref{eq:Bi}.
\end{corollary}
\begin{proof}
Let $q \coloneqq V_i(T_j)$. Since $\dot V_i=-\tilde{u}_i^2\le0$ for $t\ge T_j$, $V_i(t)\le q$. By eqn.~\eqref{eq:V}, $V_i\leqslant q$ gives $\|\chi_i\|\leqslant \sqrt{2R_iq/\kappa_{i1}}$ and $\eta_i^2-z_i^2\geqslant \eta_i^2e^{-q/\kappa_{i2}}$, hence $|z_i|\leqslant \eta_i\sqrt{1-e^{-q/\kappa_{i2}}}$. Since $u_i-\bar{\omega}=\kappa_{i1}\chi_{i,x}+\Pi_i'(z_i)$ for $t\ge T_j$, the bound follows from eqn.~\eqref{eq:Pi}.
\end{proof}
\begin{remark}[Bounded variant]
\label{rem:barrier_role}
Replacing $\Pi_i'(z_1)$ in eqn.~\eqref{eq:ui} by the bounded term $\kappa_{i2}\sin z_i$ gives the same identity $\dot{\tilde{V}}_i=-\tilde u_i^2$ with $\tilde{V}_i=\frac{\kappa_{i1}}{2R_i}\|\chi_i\|^2+\kappa_{i2}(1-\cos z_i)$ on $\mathbb R^2\times\mathbb S^1$. LaSalle then yields convergence to $\{(0,0)\}\cup\{(0,\pi)\}$, and linearising at $(0,\pi)$ gives the characteristic polynomial $s^3+(R_i\kappa_{i1}-\kappa_{i2})s^2+\bar\omega^2s-\bar\omega^2\kappa_{i2}$, whose constant term is negative, so that the antipodal equilibrium is unstable and the result is almost global. The barrier of eqn.~\eqref{eq:Pi} removes this equilibrium from the domain and delivers asymptotic stability on the whole of $\mathcal D_i$, at the cost of an input that grows as $|z_1|\to\eta_i$.
\end{remark}

Theorem~\ref{thm:follower} needs an out-neighbour that converges in finite time, which holds for the followers of leaders by Theorem~\ref{thm:leader_guidance}. A follower further down a path follows an agent that converges only asymptotically, so Assumption~\ref{ass:conv} fails for it. Using the exponential stability of Corollary~\ref{cor:les}, the next result extends Theorem~\ref{thm:follower} to the whole graph, locally. Although exponential stability implies asymptotic stability, an exponentially stable system is conventionally considered to have essentially settled after approximately five time constants. Let $T_c\geqslant 0$ be such that $c_\ell(t)=0$ for all $t\geqslant T_c$ and $\ell\in\mathcal{L}$; by Theorem~\ref{thm:leader_guidance}, $T_c=\max_{\ell\in\mathcal{L}}T_\ell$ qualifies.
\begin{Proposition}[Propagation along the graph]
\label{prop:graph}
Let $\mathcal{G}(t)$ be constant and satisfy Assumption~\ref{assump:infor_flow} on $[T_c,\infty)$, and let every follower apply guidance law given in eqn.~\eqref{eq:ui}. Then there is a neighbourhood $\mathcal{U}$ of the origin of $(\chi_i,z_i)_{i\in\mathcal{F}}$ such that, if the follower errors lie in $\mathcal{U}$ at $t=T_c$, then $r_i(t)\to R_i$, $r_{ij}(t)\to\bar r_{ij}$ with $j=j_i(T_c)$, and $\dot\gamma_i(t)\to\bar{\omega}$ for every follower $i$.
\end{Proposition}
\begin{proof}
See Appendix~\ref{app:prop}.
\end{proof}

So far, the graph has been fixed. Under the nearest-neighbour rule in eqn.~\eqref{eq:nn}, however, out-neighbours change as the agents move, and agents may join or leave; we address these changes next

\section{Extension to time-varying graphs}
\label{sec:switch}
Under the rule of eqn.~\eqref{eq:nn}, a follower changes its out-neighbour whenever another agent becomes closer, including the graph changes when agents join or leave. By Assumption~\ref{assump:switching_graph}, these changes stop after $t_K$. Two issues arise. First, at a switching instant $z_i$ is recomputed with respect to the new out-neighbour and jumps, whereas the guidance law proposed in eqn.~\eqref{eq:ui} requires $|z_i|<\eta_i$. Second, Theorem~\ref{thm:follower} requires a fixed out-neighbour that has converged, which holds only after $t_K$ and after the leaders have converged. The following assumption addresses the first issue.
\begin{assumption}[Switching instants]
\label{ass:reset}
At each $t_k$, $k\geqslant 1$, every follower $i$ that joins at $t_k$, or whose out-neighbour changes at $t_k$, satisfies $|z_i(t_k)|<\eta_i$, where $z_i$ is computed with respect to its new out-neighbour.
\end{assumption}
%
%
\begin{figure*}[ht]
     \centering
     \begin{subfigure}[b]{0.24\textwidth}
        \centering    
        \includegraphics[width=0.8\linewidth]{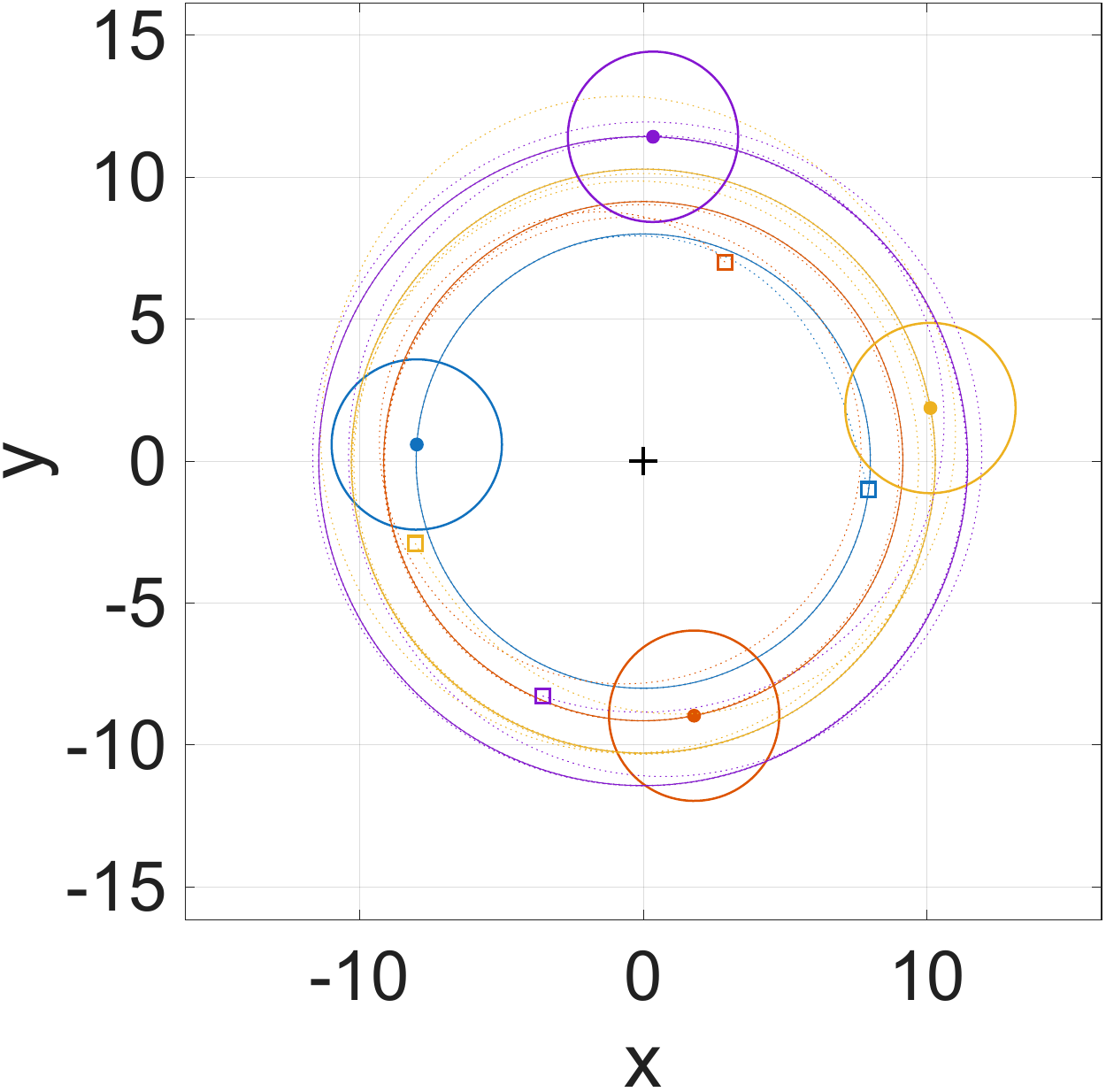}
        \caption{Trajectories}
        \label{fig:Trajectories_case_2.pdf}
     \end{subfigure}
     \begin{subfigure}[b]{0.24\textwidth}
        \centering
        \includegraphics[width=0.99\linewidth]{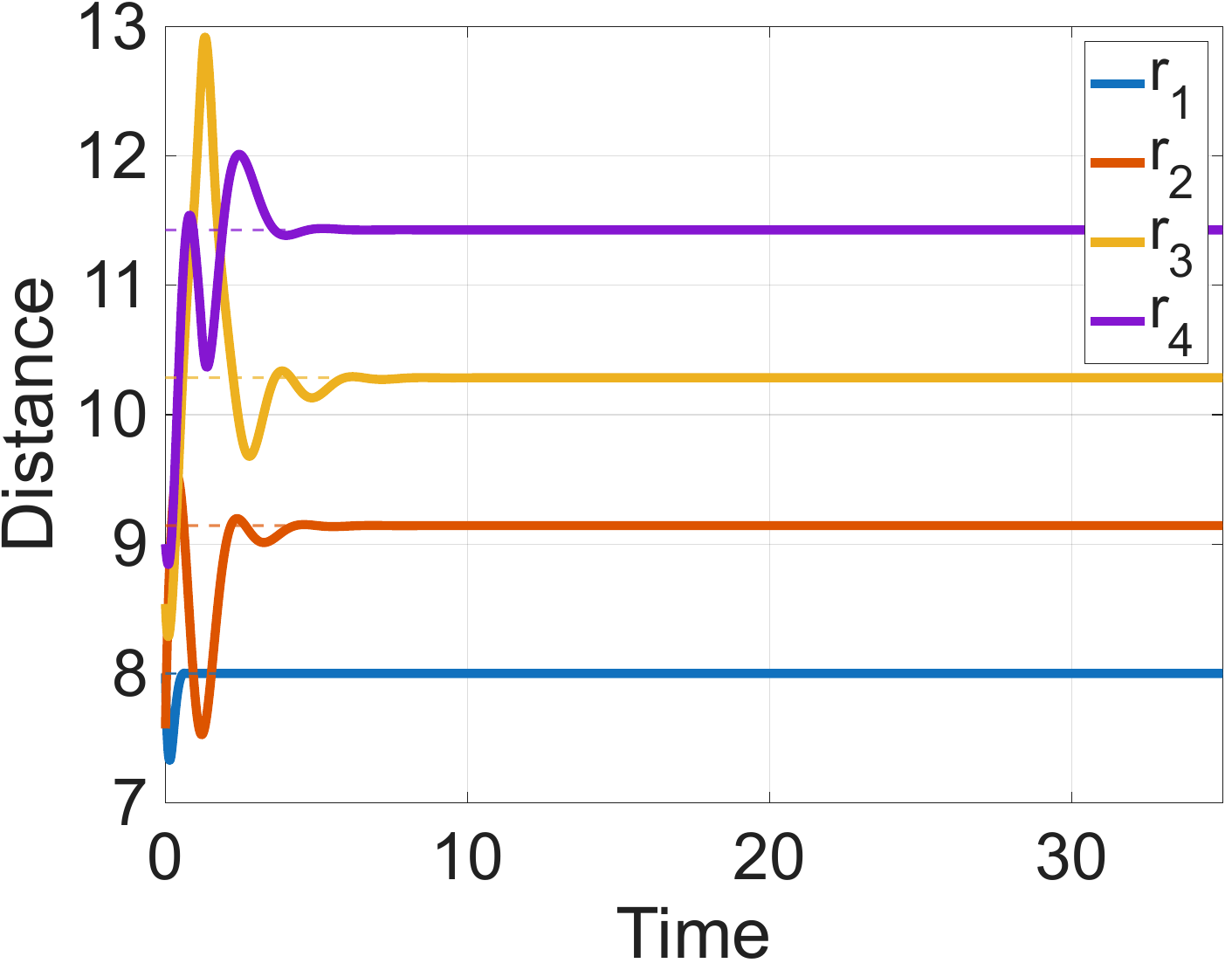} 
        \caption{Distances from the target}
        \label{fig:Distance_from_target_case_2.pdf}
     \end{subfigure}
     \begin{subfigure}[b]{0.24\textwidth}
        \centering    
        \includegraphics[width=0.96\linewidth]{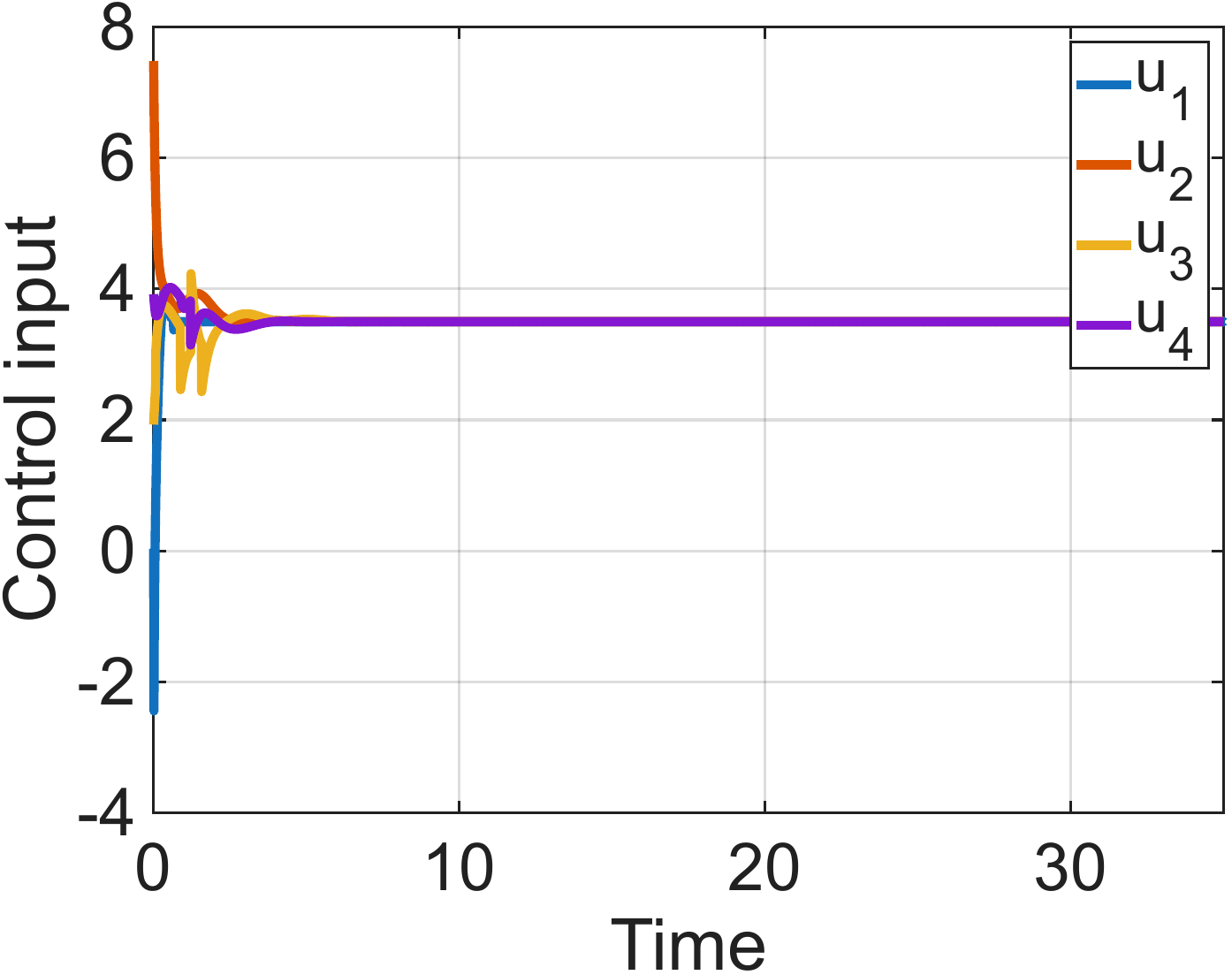}
        \caption{Control inputs}
        \label{fig:Control_inputs_case_2.pdf}
     \end{subfigure}
     \begin{subfigure}[b]{0.24\textwidth}
    \centering 
    \includegraphics[width=0.97\linewidth]{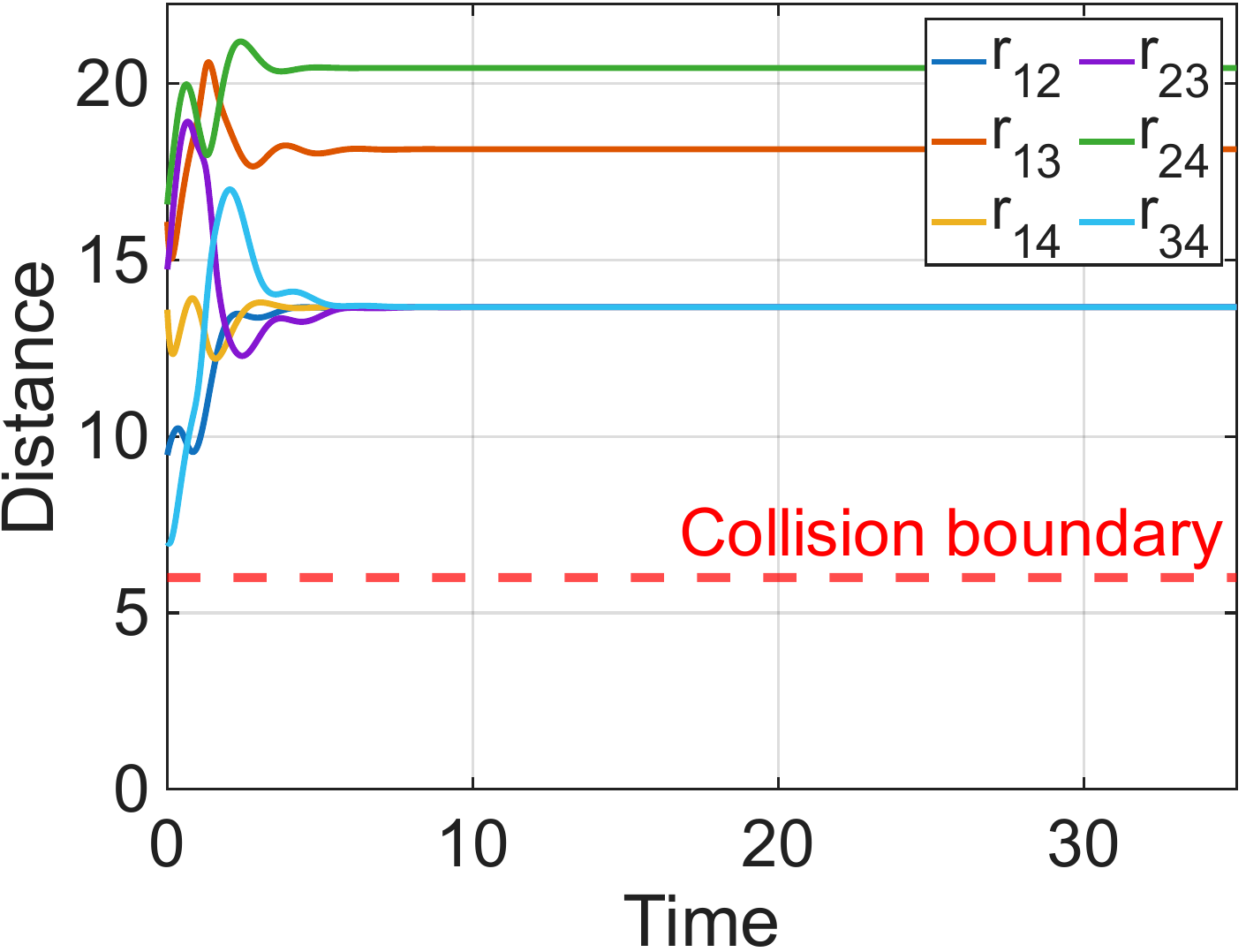}
    \caption{Inter-agent distances}
    \label{fig:Interagent_distances_case_2.pdf}
     \end{subfigure}
     \caption{Case $1$ (small circles in subfigure (a) represent safety radius)}
     \label{fig:Case_2_combined}
\end{figure*}
\begin{figure*}[ht]
     \centering
     \begin{subfigure}[b]{0.24\textwidth}
    \centering 
    \includegraphics[width=0.8\linewidth]{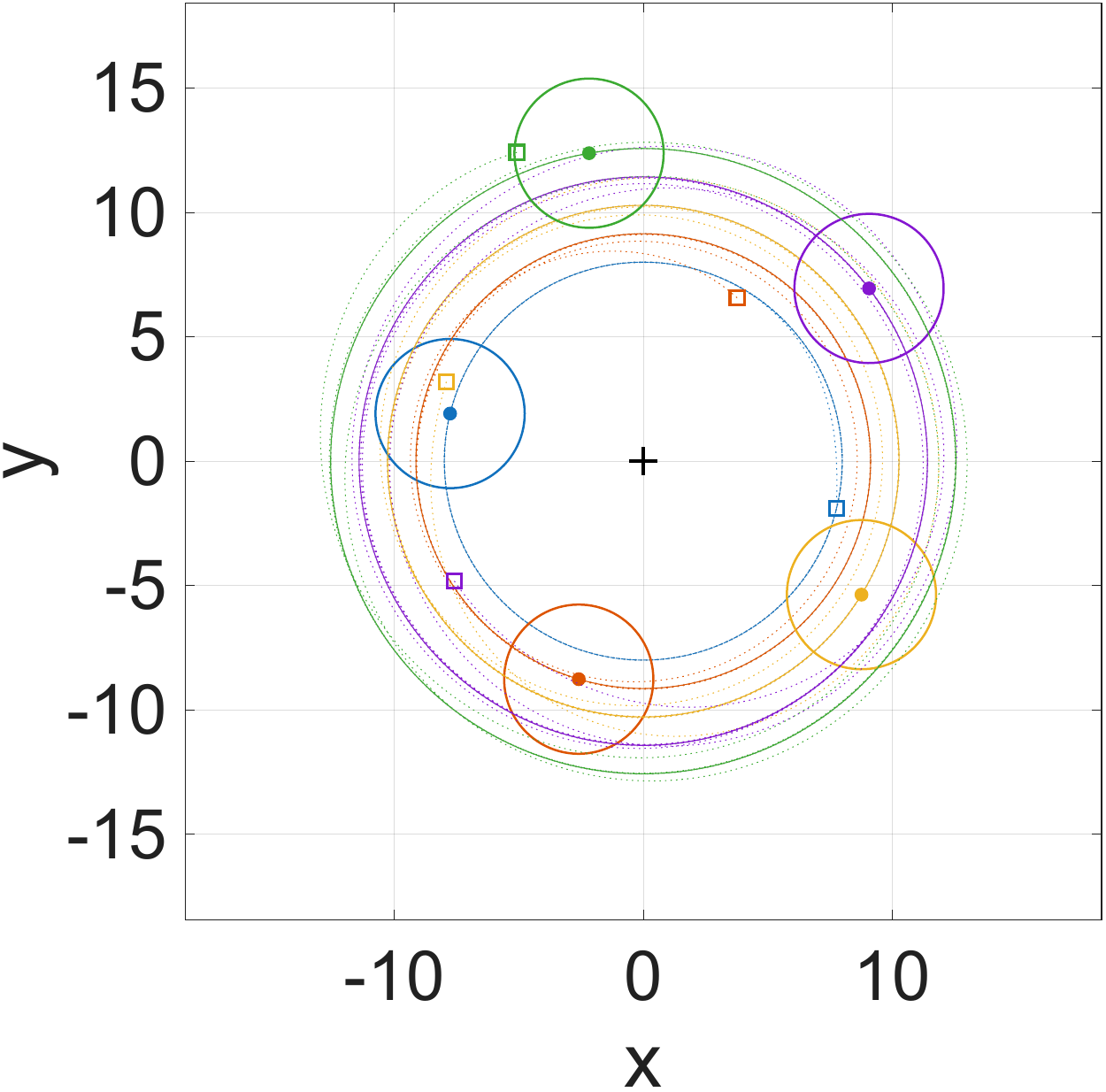}
    \caption{Trajectories}
    \label{fig:Trajectories_case_3.pdf}
     \end{subfigure}
     \begin{subfigure}[b]{0.24\textwidth}
        \centering
        \includegraphics[width=0.97\linewidth]{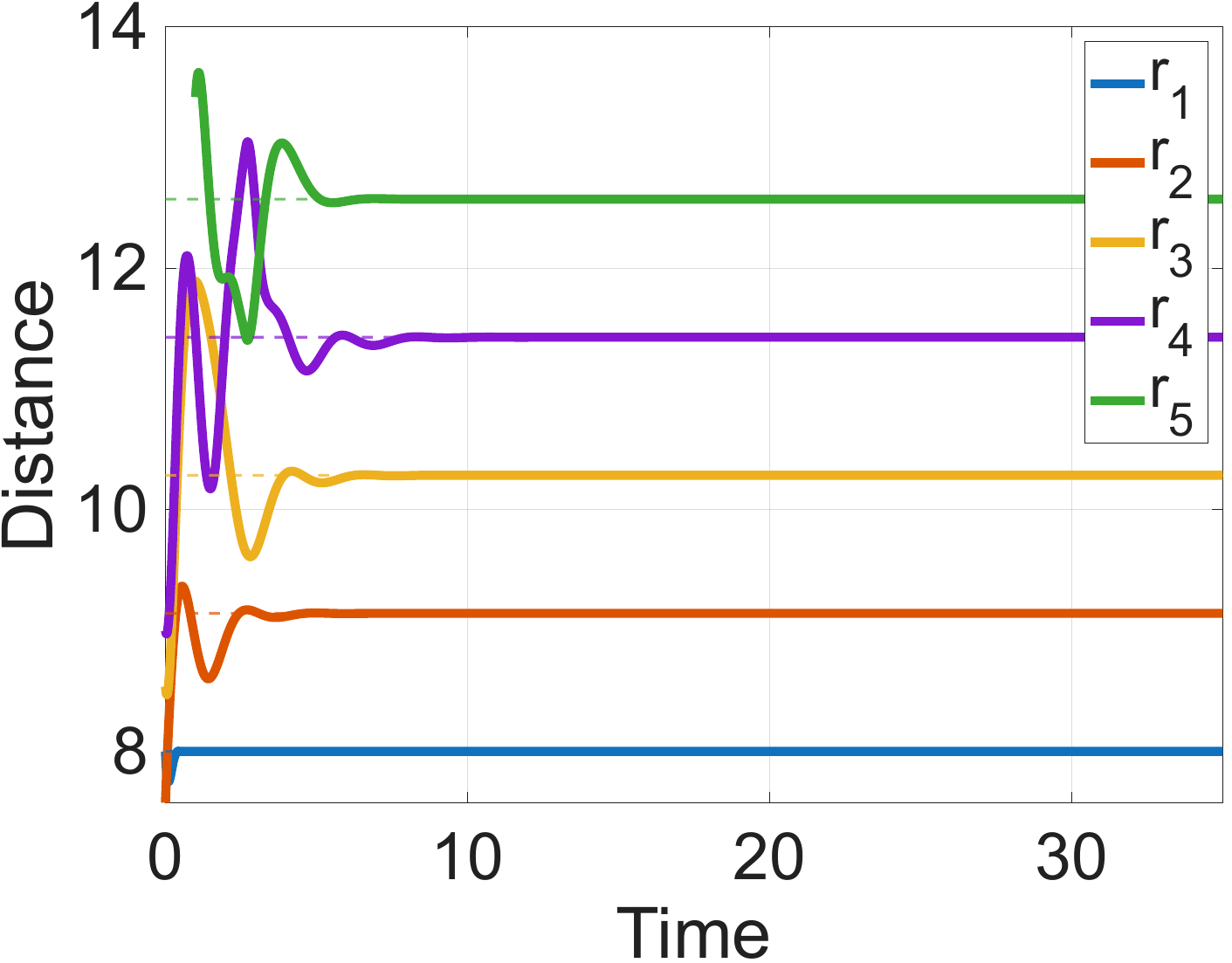} 
        \caption{Distances from the target}
        \label{fig:Distance_from_the_target_case_3.pdf}
     \end{subfigure}
     \begin{subfigure}[b]{0.24\textwidth}
        \centering    
        \includegraphics[width=0.95\linewidth]{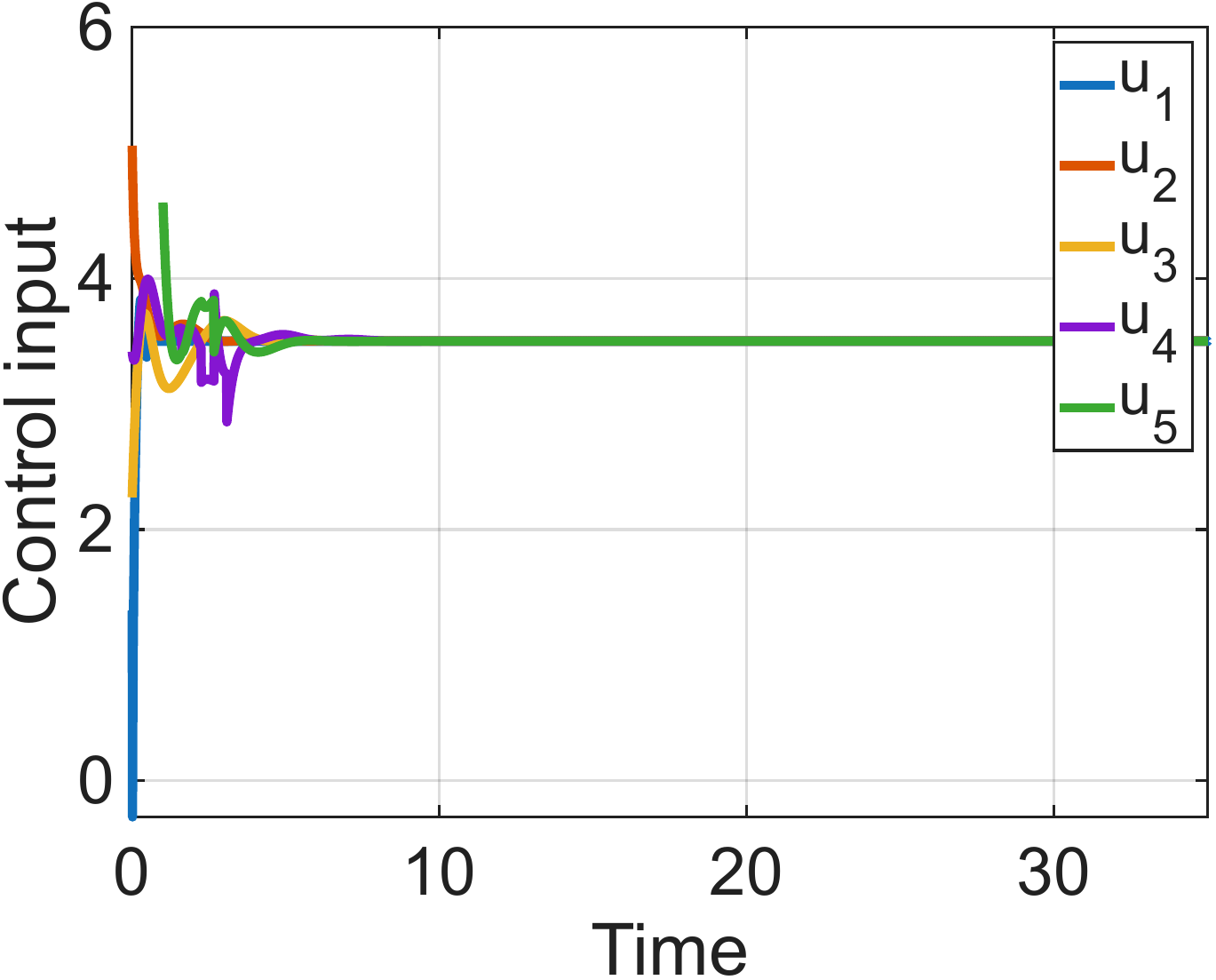}
        \caption{Control inputs}
        \label{fig:Control_inputs_case_3.pdf}
     \end{subfigure}
     \begin{subfigure}[b]{0.24\textwidth}
    \centering 
    \includegraphics[width=0.97\linewidth]{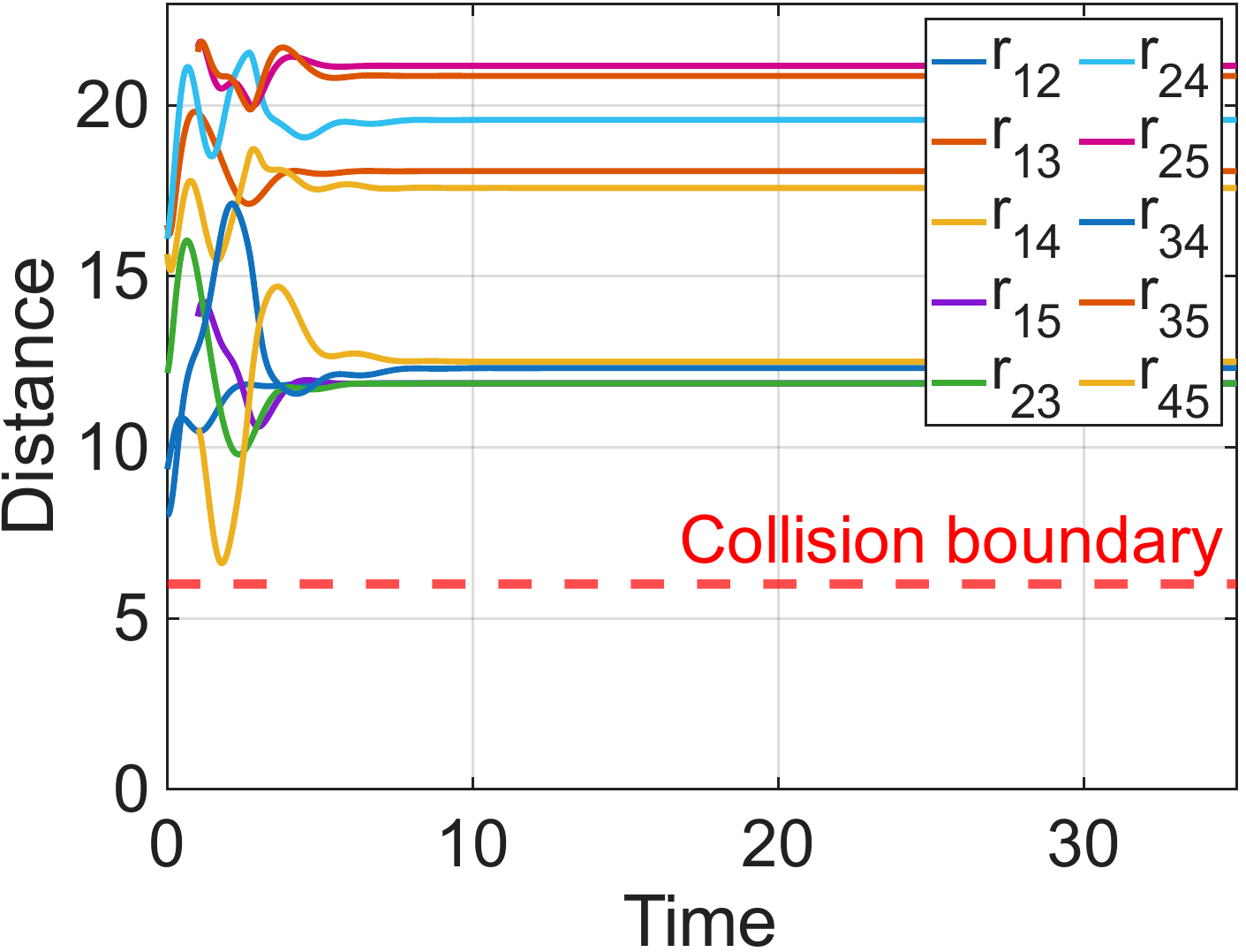}
    \caption{Inter-agent distances}
    \label{fig:Inter-agent_distances_case_3.pdf}
     \end{subfigure}
     \caption{Case $2$ (small circles in subfigure (a) represent safety radius)}
     \label{fig:Case_3_combined}
\end{figure*}
\begin{figure*}[ht]
     \centering
     \begin{subfigure}[b]{0.24\textwidth}
    \centering 
    \includegraphics[width=0.8\linewidth]{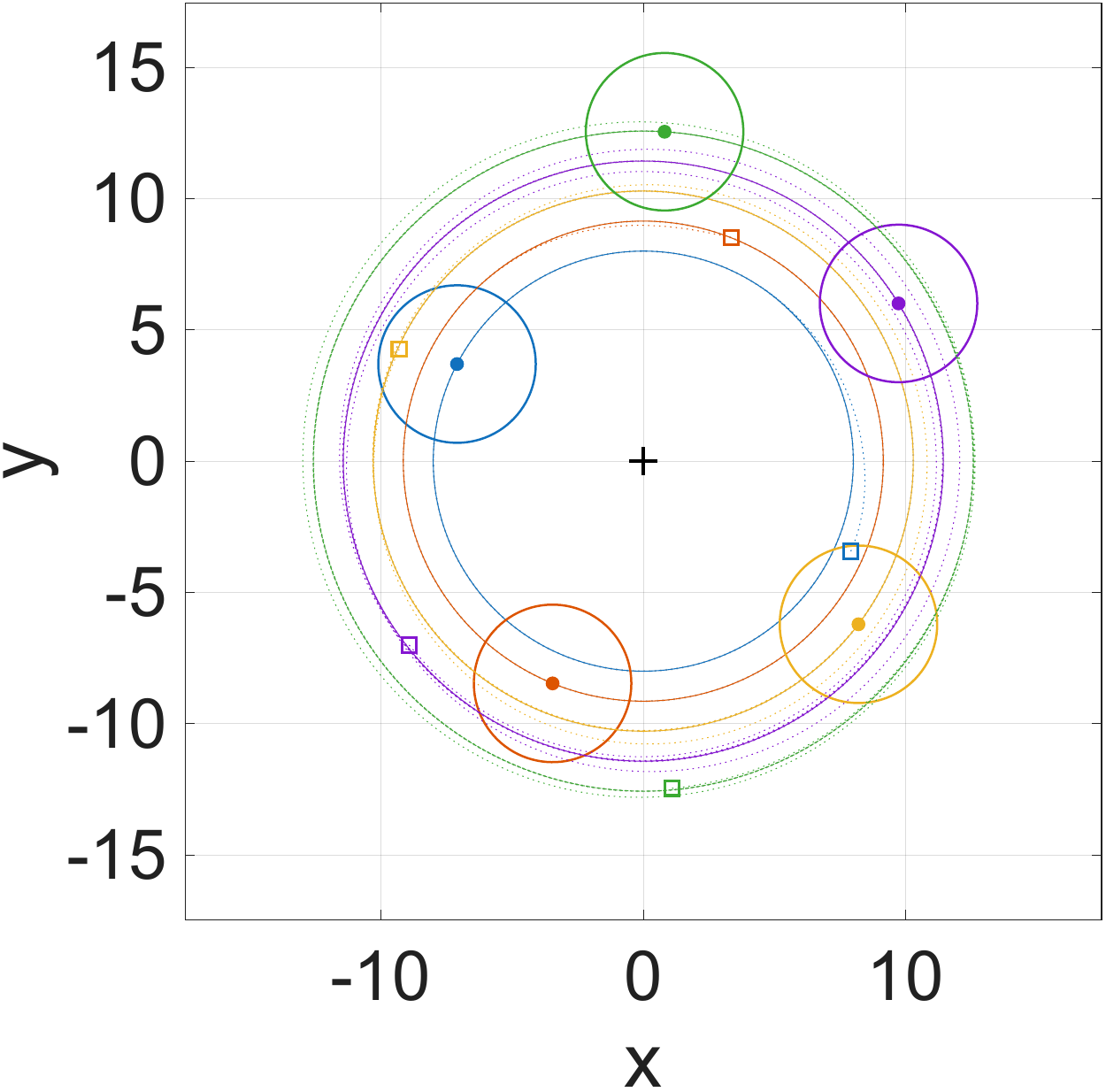}
    \caption{Trajectories}
    \label{fig:Trajectories_case_4.pdf}
     \end{subfigure}
     \begin{subfigure}[b]{0.24\textwidth}
        \centering
        \includegraphics[width=0.97\linewidth]{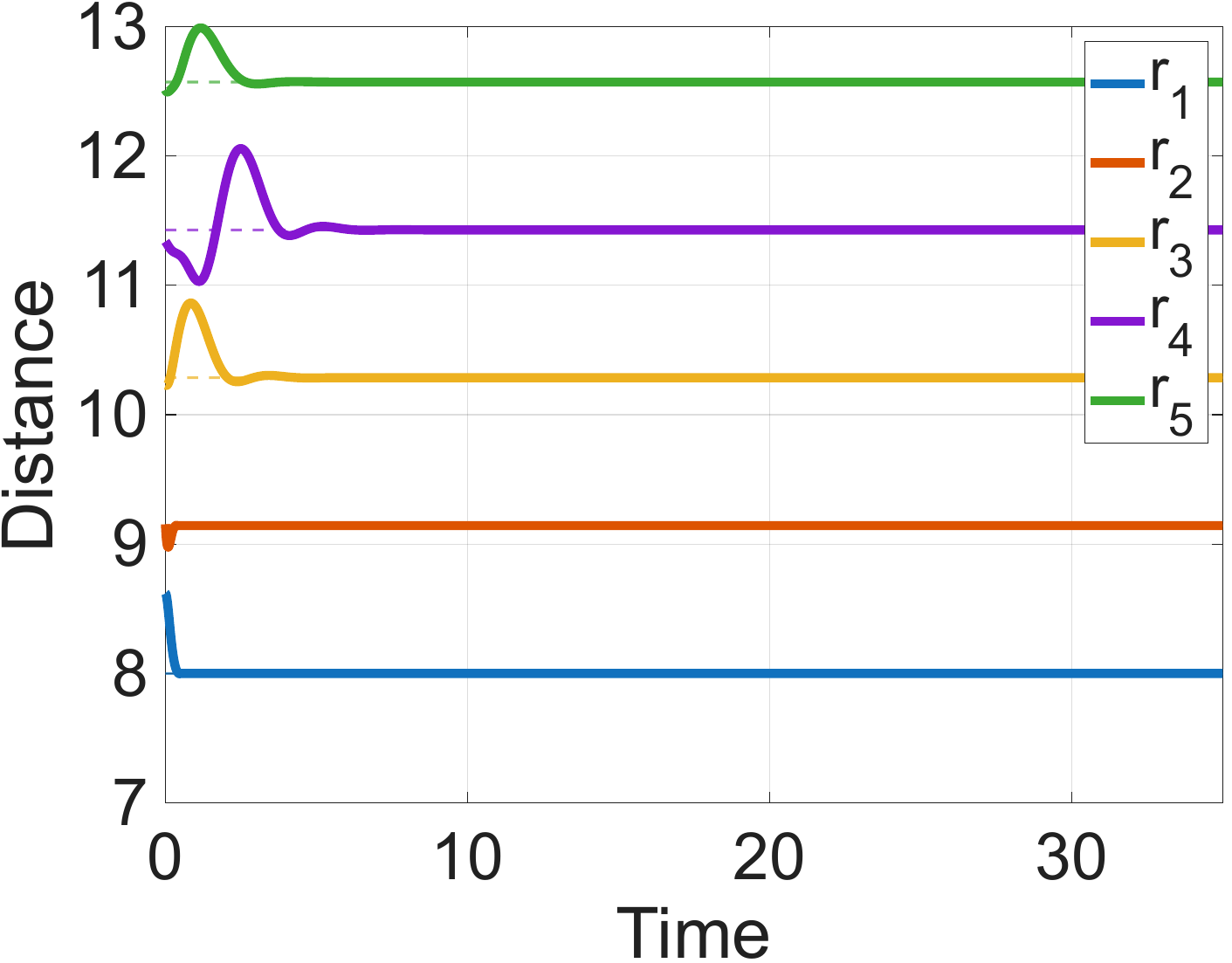} 
        \caption{Distances from the target}
        \label{fig:Distance_from_the_target_case_4.pdf}
     \end{subfigure}
     \begin{subfigure}[b]{0.24\textwidth}
        \centering    
        \includegraphics[width=0.95\linewidth]{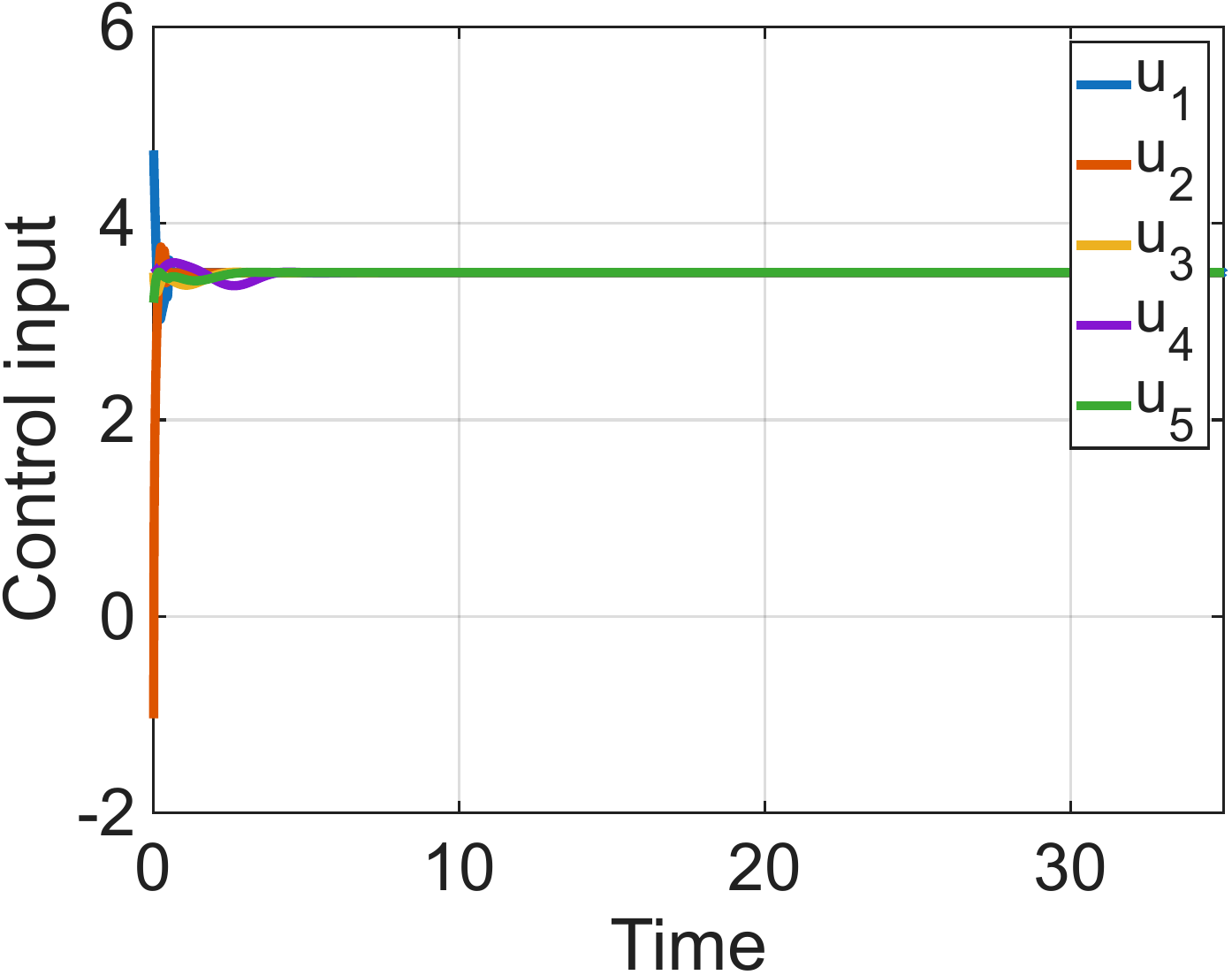}
        \caption{Control inputs}
        \label{fig:Control_inputs_case_4.pdf}
     \end{subfigure}
     \begin{subfigure}[b]{0.24\textwidth}
    \centering 
    \includegraphics[width=0.97\linewidth]{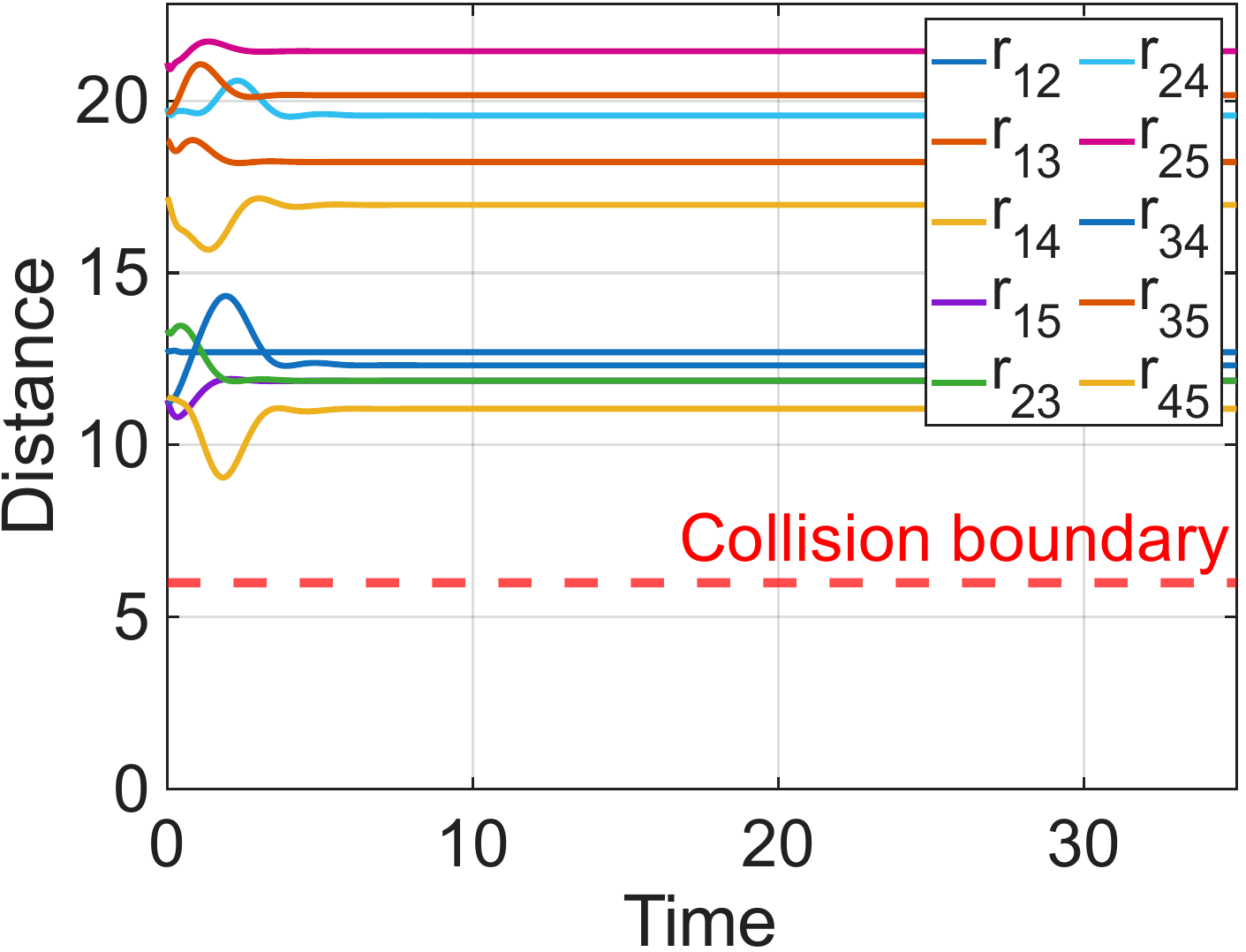}
    \caption{Inter-agent distances}
    \label{fig:Inter-agent_distances_case_4.pdf}
     \end{subfigure}
     \caption{Case $3$ (small circles in subfigure (a) represent safety radius)}
     \label{fig:Case_4_combined}
\end{figure*}

Assumption~\ref{ass:reset} constrains the state only at the finitely many switching instants; between them, Lemma~\ref{lem:barrier} keeps every $z_i$ inside its barrier. Let $T^\ast:=\max\{t_K,\ \max_{\ell\in\mathcal{L}}T_\ell\}$. From $T^\ast$ on, the graph is fixed and every leader has converged, which resolves the second issue. The following theorem combines Theorem~\ref{thm:leader_guidance} for the leaders, Lemma~\ref{lem:barrier} on $[0,T^\ast]$, and Theorem~\ref{thm:follower} and Proposition~\ref{prop:graph} on $[T^\ast,\infty)$, and states which requirements of the Problem hold under which conditions.
\begin{thm}
\label{thm:main}
Consider a group of $n$ agents in which Assumptions~\ref{assump:infor_flow}, \ref{assump:switching_graph} and~\ref{ass:reset} hold. Let every leader satisfy the hypotheses of Theorem~\ref{thm:leader_guidance}, with $u_\ell$ bounded on $[0,T^\ast]$; and let every follower apply guidance law given in eqn.~\eqref{eq:ui}, with $|z_i(0)|<\eta_i$ for the followers present at $t=0$. Then the solution exists for all $t\geqslant 0$, and Problem is solved in the following sense:
\begin{enumerate}
\item[(a)] (P1) holds, and the leader part of (P3) holds if $\dist(I_\ell,I_m)>2R_s$ for all distinct $\ell,m\in\mathcal{L}$;
\item[(b)] every follower $i$ whose out-neighbour on $[t_K,\infty)$ is a leader satisfies (P2) and the follower part of (P3), the latter for all $t\geqslant T^\ast$ if $V_i(T^\ast)<\bar q_i$;
\item[(c)] if the errors of all followers at $t=T^\ast$ lie in a neighbourhood $\mathcal{U}$ as in Proposition~\ref{prop:graph}, then every follower satisfies (P2) and the follower part of (P3).
\end{enumerate}
\end{thm}
\begin{proof}
On each interval $[t_k,t_{k+1})$, with $t_{K+1} \to \infty$, the graph is constant, and by Assumption~\ref{assump:infor_flow} the followers can be ordered by depth. On any bounded subinterval $[a,b]$, every leader input is bounded: by hypothesis on $[0,T^\ast]$, and because $u_\ell=\bar{\omega}$ for $t\geqslant T_\ell$ by Theorem~\ref{thm:leader_guidance} and Lemma~\ref{lem:tc}(a). Lemma~\ref{lem:barrier}, applied in order of increasing depth, then gives existence on $[a,b]$, keeps every $z_i$ in $(-\eta_i,\eta_i)$, and bounds every $u_i$. At each $t_k$, Assumption~\ref{ass:reset} gives $|z_i(t_k)|<\eta_i$ for the followers that join or change out-neighbour. Induction over $k$ gives existence for all $t\geqslant 0$, and in particular $|z_i(T^\ast)|<\eta_i$.

(a) This is Theorem~\ref{thm:leader_guidance}, since $c_\ell(t)=0$ gives $r_\ell(t)=R_\ell$ by eqn.~\eqref{eq:turn_centre}. On $[T^\ast,\infty)$, the graph is constant and satisfies Assumption~\ref{assump:infor_flow}, and $c_\ell\equiv0$ for every leader. (b) For such a follower, Assumption~\ref{ass:conv} holds with $T_j=T^\ast$, so Theorem~\ref{thm:follower} gives (P2) and $r_{ij}(t)\to\bar r_{ij}>2R_s$, hence the follower part of (P3); Theorem~\ref{thm:follower}(c) gives it for all $t\ge T^\ast$ if $V_i(T^\ast)<\bar q_i$. (c) Proposition~\ref{prop:graph} with $T_c=T^\ast$ gives (P2) and $r_{ij}(t)\to\bar r_{ij}>2R_s$ for every follower. Hence, proved.
\end{proof}

Corollary~\ref{cor:leader-input} gives sufficient conditions for the bounded leader inputs that Theorem~\ref{thm:main} assumes. The theorem guarantees separation from the out-neighbour for all $t$.

\section{Simulation results}
\label{Sec:Simulaton}
We illustrate the results of Sections~\ref{sec:guidance_for_leaders}--\ref{sec:switch} in three cases of increasing complexity, with the target at the origin: first, for for a time-varying interaction topology with a fixed number of nodes; second, for a time-varying interaction topology with a changing number of nodes; and third, for a static topology with two leaders. In all interaction graphs, leader and follower agents are shown in red and green, respectively. In every case the desired radii are set as $R_k=v_k/\bar\omega$ in accordance with eqn.~\eqref{eq:R_k}, and the formation parameters $\bar e_{i_1}^j$ are chosen so that eqn.~\eqref{eq:coscond} holds for every edge of the interaction graph.

\begin{figure}[ht]
    \centering
        \begin{subfigure}[b]{0.22\textwidth}
        \centering    
        \includegraphics[width=\linewidth]{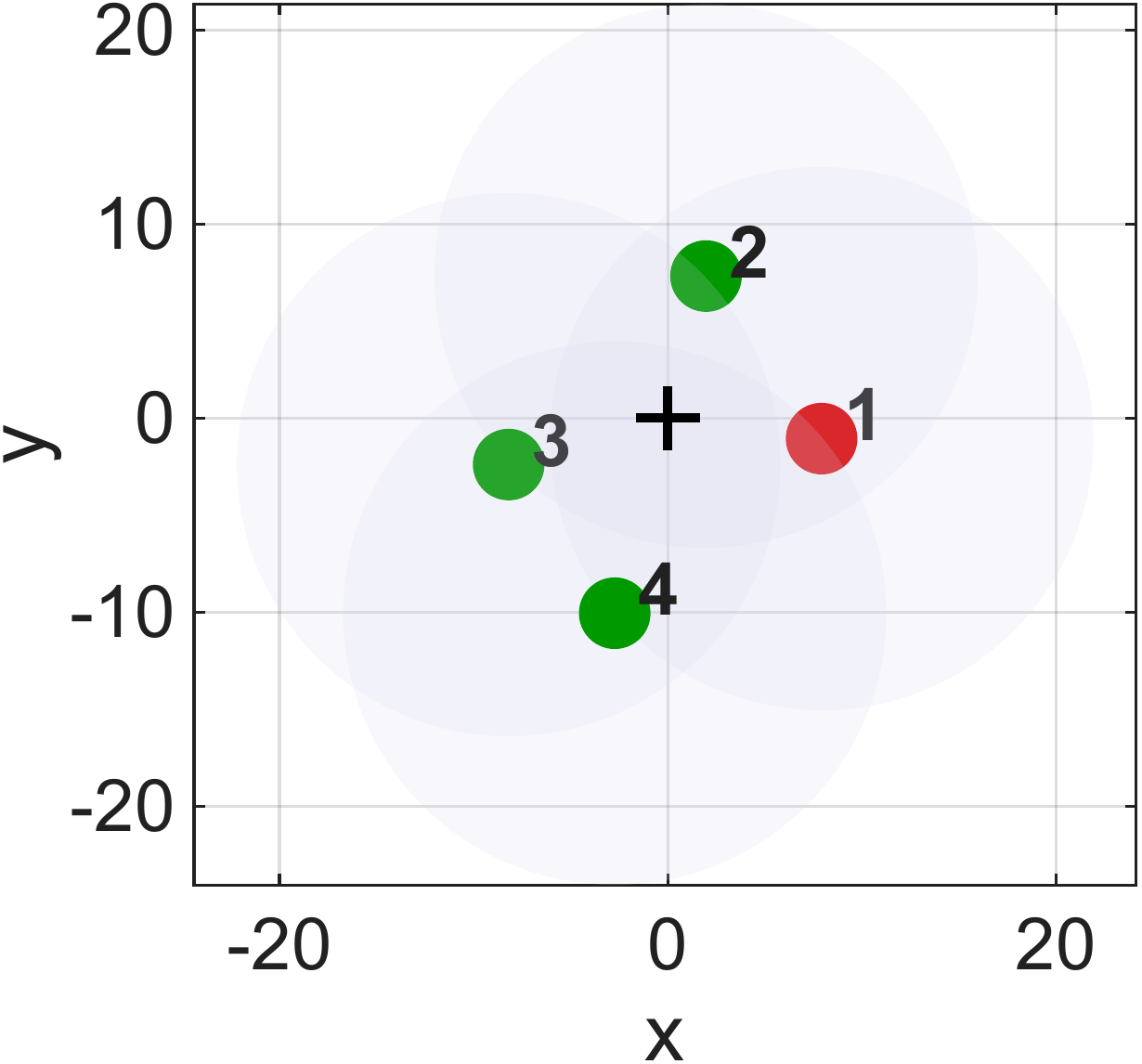}
        \caption{Distribution in $\mathbb{R}^2$}
        \label{fig:sensing_graph_sim}
    \end{subfigure}
    \hfill
    \begin{subfigure}[b]{0.22\textwidth}
        \centering
        \includegraphics[width=\linewidth]{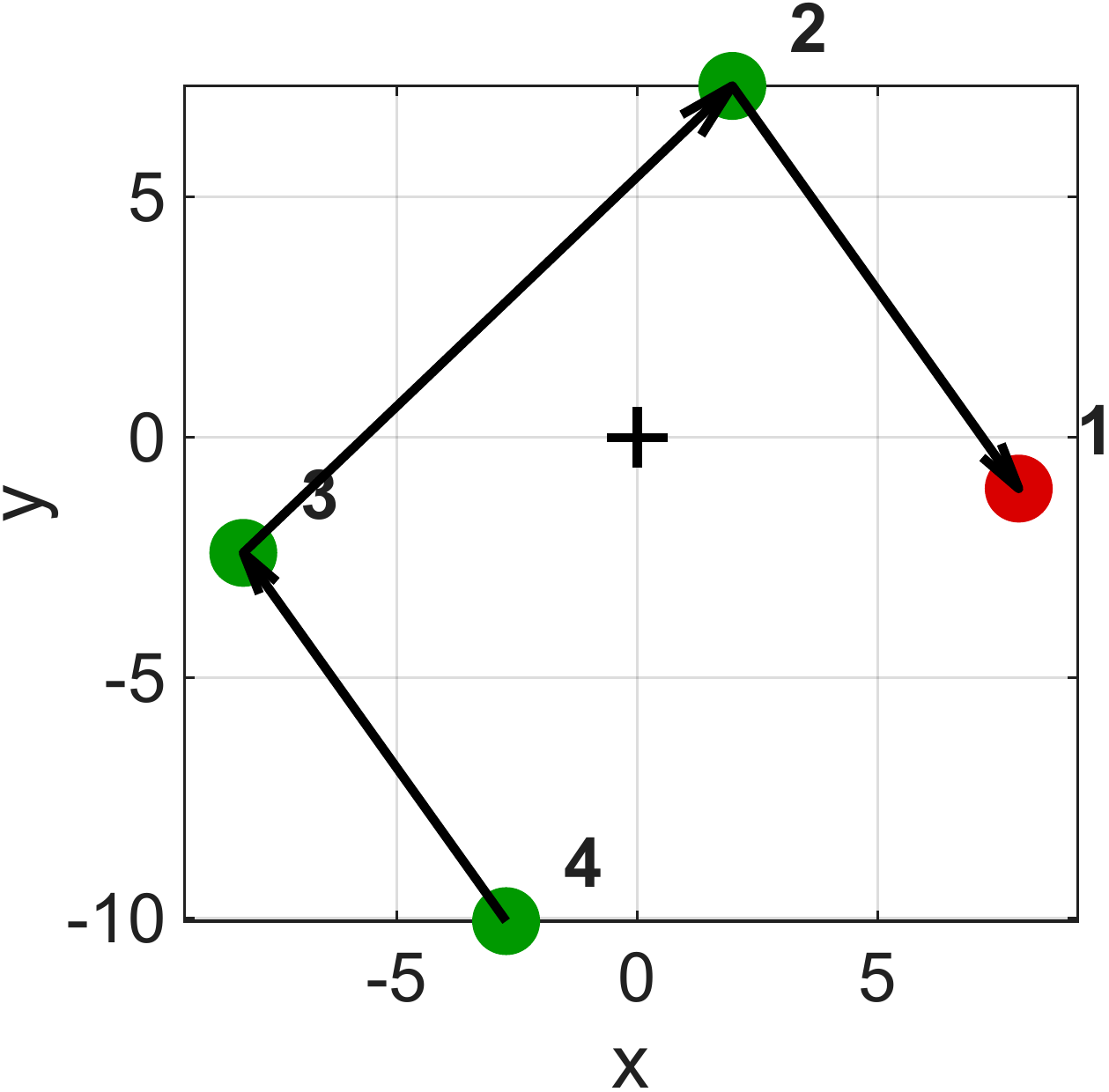} 
        \caption{Interaction graph $\mathcal{G}$}
        \label{fig:communication_graph_sim}
    \end{subfigure}
    \caption{Static interaction graph based on nearest neighbour selection rule at $t=0$. Leaders are shown in red and followers in green.}
    \label{fig:sense_to_communication_sim}
\end{figure}
%


%

\begin{figure*}[ht]
    \centering
    \includegraphics[width=0.8\linewidth]{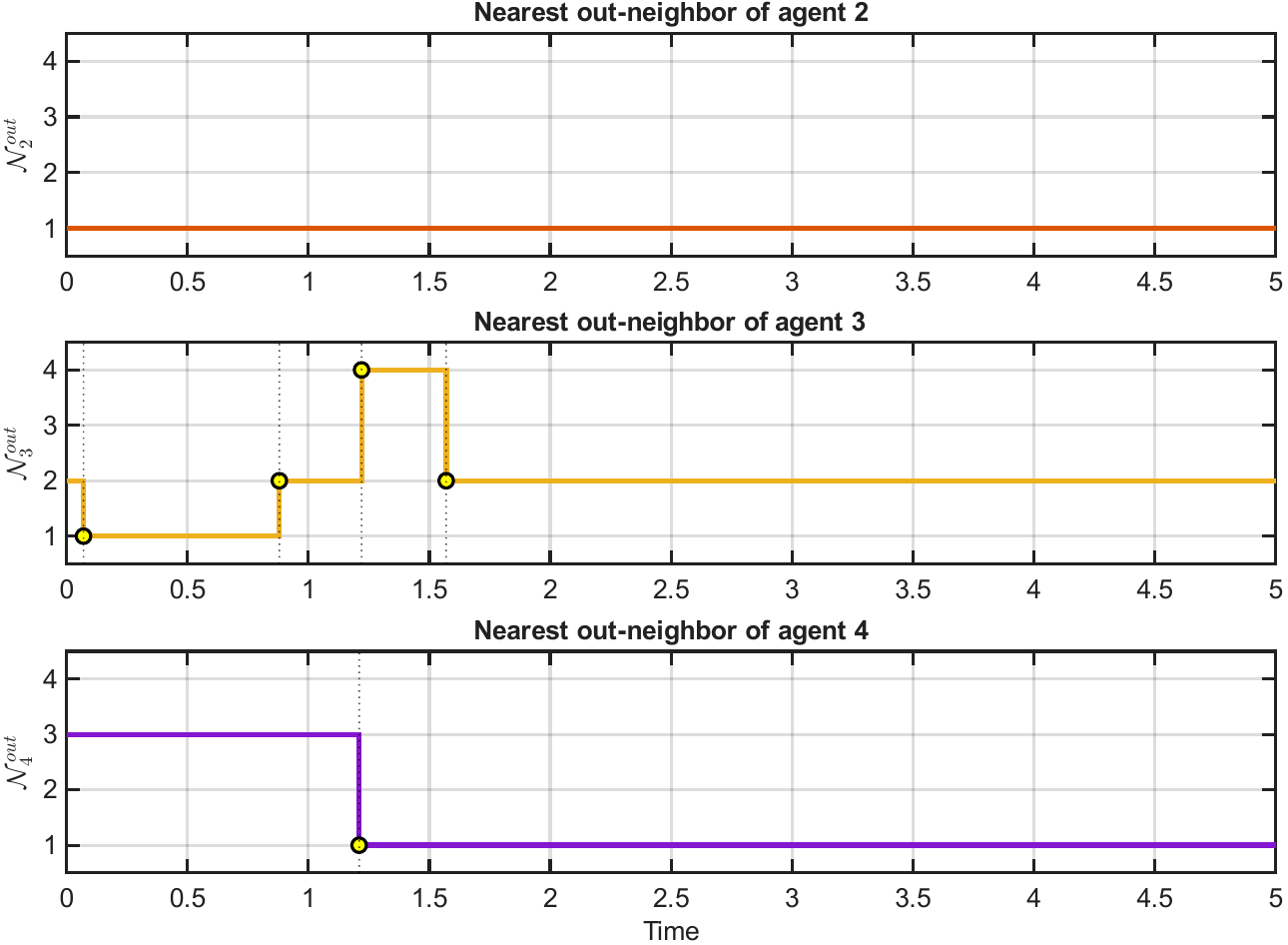}
    \caption{Nearest out-neighbours of all the followers till $t=5$ s for case $2$}
    \label{fig:out_neighbour_switches_case_2}
\end{figure*}
\begin{figure*}[ht]
    \centering
    \includegraphics[width=0.8\linewidth]{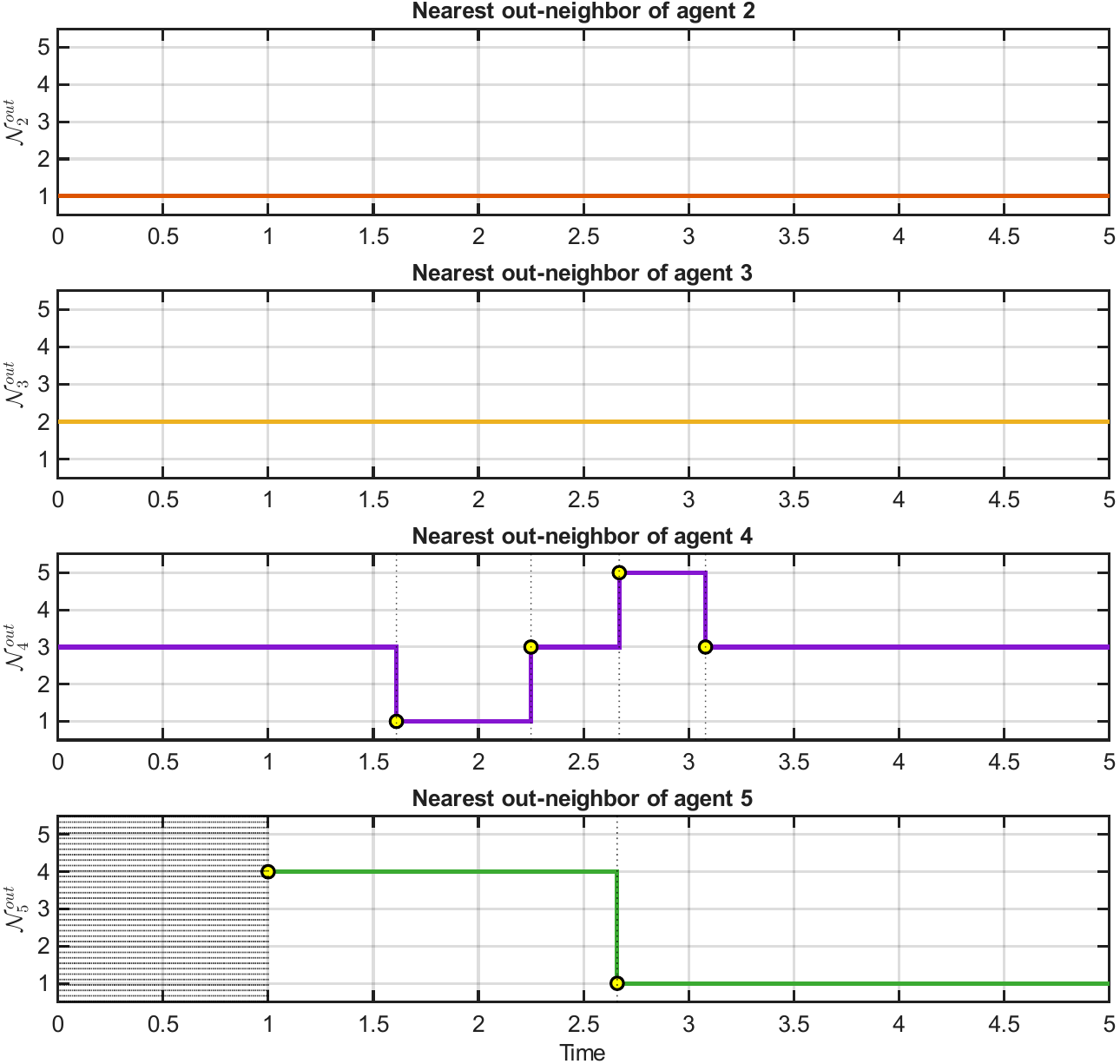}
    \caption{Nearest out-neighbours of all the followers till $t=5$ s for case $3$}
    \label{fig:out_neighbour_switches_case_3.pdf}
\end{figure*}
%

\subsection*{Case $1$: Time-varying interaction topology with fixed number of nodes}
\begin{table}[h]
\caption{Simulation's initial conditions for case $1$ (in SI units)}
\centering
\begin{tabular}{|c|c|c|c|c|c|c|}
\hline
Agent  & $P$ & $v$ & $\gamma$ & $R_s$ & Desired radii \\
\hline
1 & (7.93,-0.99)     & 28 & 1.9 & 3 & 8 \\
\hline
2 & (2.88,7) & 32 & 2.55 & 3 & 9.14 \\
\hline
3 & (-8.03,-2.89)  & 36 & -1.05  & 3 & 10.29 \\
\hline
4 & (-3.54,-8.27)  & 40 & -0.33 & 3 & 11.43 \\
\hline
\end{tabular}
\label{tab_case2}
\end{table}
This case illustrates the time-varying graph results with fixed number of nodes: the initial conditions are such that the nearest-neighbour rule switches the graph during the transient, as allowed by Assumption~\ref{assump:switching_graph}. The parameters are as shown in the initial conditions in Table~\ref{tab_case2}. To handle a more challenging scenario, each in-neighbour is faster than its out-neighbour. Also, the linear speeds are kept such that their final paths are sufficiently close, which makes the design condition of eqn.~\eqref{eq:coscond} the binding safety requirement.
Under the nearest-neighbour selection rule, the interaction topology evolves over time, leading to switching in the set of out-neighbours. The corresponding switching behaviour is illustrated in Fig.~\ref{fig:out_neighbour_switches_case_2}, where the plot is truncated at $t=5$ s since no further changes occur.

Despite the time-varying nature of the interaction graph, all follower agents achieve synchronisation with the leader’s angular speed and converge to concentric circular trajectories around the target. The resulting trajectories are shown in Fig.~\ref{fig:Trajectories_case_2.pdf}, while Fig.~\ref{fig:Distance_from_target_case_2.pdf} shows convergence of radial distances to constant steady-state values.

The control inputs, depicted in Fig.~\ref{fig:Control_inputs_case_2.pdf}, stabilise over time, indicating convergence to a common angular speed. In Fig.~\ref{fig:Interagent_distances_case_2.pdf}, all inter-agent distances remain strictly above the prescribed threshold of $6$.

In this case the set of agents is fixed. Case~$2$ adds an agent at $t=1$~s, which creates new edges and illustrates the joining events covered by Assumption~\ref{assump:switching_graph} and Theorem~\ref{thm:main}.

\subsection*{Case $2$: Time-varying interaction topology with a varying number of nodes}
\begin{table}[h]
\caption{Simulation's initial conditions for case $2$ (in SI units)}
\centering
\begin{tabular}{|c|c|c|c|c|c|c|}
\hline
Agent  & $P$ & $v$ & $\gamma$ & $R_s$ & Desired radii \\
\hline
1 & (-7.77,-1.9)     & 28 & -1.56 & 3 & 8 \\
\hline
2 & (3.77,6.57) & 32 & 2.39 & 3 & 9.14 \\
\hline
3 & (-7.92,3.19)  & 36 & -1.87  & 3 & 10.29 \\
\hline
4 & (-7.6,-4.82)  & 40 & -0.97& 3 & 11.43 \\
\hline
5 & (-5.08,12.41)  &   44  & -2.86 &  3   & 12.57\\
\hline
\end{tabular}
\label{tab_case3}
\end{table}
The initial conditions are mentioned in Table~\ref{tab_case3}. A new agent is introduced into the system at $t=1$ s, resulting in a dynamically evolving node set alongside the time-varying interaction topology. The evolution of out-neighbour relationships is illustrated in Fig.~\ref{fig:out_neighbour_switches_case_3.pdf}, with switching activity ceasing after approximately $t=5$ s.

All agents synchronise with the leader’s angular speed and converge to concentric circular paths about the target, as shown in Fig.~\ref{fig:Trajectories_case_3.pdf}. The radial distances, plotted in Fig.~\ref{fig:Distance_from_the_target_case_3.pdf}, approach constant values.

The control inputs (Fig.~\ref{fig:Control_inputs_case_3.pdf}) converge to steady-state values, further indicating uniform rotational motion. In Fig.~\ref{fig:Inter-agent_distances_case_3.pdf}, all pairwise distances remain above the threshold value of $6$. The instant of agent insertion is marked by a vertical black line in the figure. 

Cases~$1$--$2$ have a single leader, which fixes all phases relative to it. Case~$3$ has two leaders.

\subsection*{Case $3$: Static interaction topology with two leaders}
\setlength{\abovedisplayskip}{2pt}
\setlength{\belowdisplayskip}{2pt}
\begin{table}[h]
\caption{Simulation's initial conditions for case $3$ (in SI units)}
\centering
\begin{tabular}{|c|c|c|c|c|c|c|}
\hline
Agent  & $P$ & $v$ & $\gamma$ & $R_s$ & Desired radii \\
\hline
1 & (7.91,-3.43)     & 28 & 1.16 & 3 & 8 \\
\hline
2 & (3.55,8.52) & 32 & 2.92 & 3 & 9.14 \\
\hline
3 & (-9.3,4.26)  & 36 & -2  & 3 & 10.29 \\
\hline
4 & (-8.92,-7)  & 40 & -0.9& 3 & 11.43 \\
\hline
5 & (1.09,-12.46)  &   44  & 0.10 &  3   & 12.57\\
\hline
\end{tabular}
\label{tab2}
\end{table}

The initial conditions for this case are listed in Table~\ref{tab2}. Unlike the previous scenarios, agents $1$ and $2$ are leaders. The simulation results show that all agents converge to concentric circular trajectories around the target, as illustrated in Fig.~\ref{fig:Trajectories_case_4.pdf}. The evolution of the radial distances from the target is shown in Fig.~\ref{fig:Distance_from_the_target_case_4.pdf}, where all distances converge to constant steady-state values.

The corresponding control inputs are depicted in Fig.~\ref{fig:Control_inputs_case_4.pdf}. Furthermore, Fig.~\ref{fig:Inter-agent_distances_case_4.pdf} shows that all inter-agent distances remain strictly greater than the safety threshold $2R_s = 6$ throughout the evolution. 

These simulations illustrate Theorems~\ref{thm:leader_guidance}, \ref{thm:follower} and \ref{thm:main}.

\section{Conclusion}
\label{Sec:Conclusion}
In this paper, we presented distributed guidance laws for the circumnavigation of a stationary target by unicycle agents with possibly different constant speeds, modelled as disks, when only the leaders know the target position. The leaders reach their circles in finite time and, under an explicit condition, never collide with each other. Each follower acts on its turn centre and, once its out-neighbour has converged, computes the component of the turn centre along its heading exactly from three angular measurements; a logarithmic barrier Lyapunov function keeps the relative heading error within prescribed bounds. For a follower whose out-neighbour has converged, the desired formation with that neighbour is locally exponentially stable for all positive gains, with a region of attraction equal to the whole set where the heading constraint holds, and an explicit sublevel-set condition keeps the follower separated from its out-neighbour. These results extend locally along the interaction graph and to switching graphs with finitely many switches, including agents joining or leaving. The follower law requires no range-rate measurement, no signal differentiation and no state memory. Simulations illustrate the results.

Future work will focus on strengthening the switching-graph analysis beyond piecewise static graphs, extending the framework to moving targets and environmental obstacles, and incorporating communication delays and sensing uncertainty into the design.
\appendix
\renewcommand\thesection{\Alph{section}}
\section{Proofs}
\label{app:proofs}
\subsection{Proof of Theorem~\ref{thm:follower}}
\label{app:thm2}
We prove the three parts in order; parts (b) and (c) use the forward invariance established in part (a).

(a) By eqn.~\eqref{eq:Vdot} and \eqref{eq:ui},
\setlength{\abovedisplayskip}{2pt}
\setlength{\belowdisplayskip}{2pt}
\begin{equation}
\label{eq:Vneg}
\dot V_i=-\tilde{u}_i\big(\kappa_{i1}\chi_{i,x}+\Pi_i'(z_i)\big)=-\tilde{u}_i^2\le0 .
\end{equation}
And since $\Pi_i(z_i)\leqslant V_i$ with $\Pi_i(z_i)\to +\infty$ as $|z_i| \to \eta_i$, the constraint $|z_i(t)|<\eta_i$ is preserved; this exactly the conclusion of Lemma~\ref{Lemma:BLF}. Since, $V_i$ also tends to $+\infty$, it is proper on $\mathcal{D}_i$, and its sublevel sets
\setlength{\abovedisplayskip}{2pt}
\setlength{\belowdisplayskip}{2pt}
\begin{equation}
    \label{eq:subl_sets}
    \Omega_q \coloneqq \{(\chi_i,z_i) \in \mathcal{D}_i: \quad V_i(\chi_i,z_i)\leqslant q\}
\end{equation}
are compact subsets of $\mathcal{D}_i$, with $\bigcup_{q>0}\Omega_q=\mathcal{D}_i$. By eqn.\eqref{eq:subl_sets}, each $\Omega_q$ is forward invariant, so solution exists and remain in $\mathcal{D}_i$ for all $t$.

(b) Setting the right hand sides of eqn.\eqref{eq:closed} to zero, $\dot{z}_i=0$ gives $\tilde{u}_i=0$ and hence $u_i=\bar{\omega}>0$; then $\dot{\chi}_{i,y}=0$ gives $\chi_{i,x}=0$, and $\dot{\chi}_{i,x}=0$ gives $\chi_{i,y}=0$; finally $\tilde{u}_i=0$ with $\chi_{i,x}=0$ forces $\Pi_i^{'}=0$ and on $(-\eta_i,\eta_i)$ the function $\Pi_i^{'}$ vanishes only at $z_i=0$.

Fix any initial condition in $\mathcal{D}_i$ and let $q \coloneqq V_i(t_0)$ at the initial time $t_0$. The closed loop in eqns.~\eqref{eq:closed},~\eqref{eq:ui} is autonomous and smooth on $\mathcal{D}_i$ for $t$, and $\Omega_q$ is compact and forward invariant by (a); LaSalle's invariance principle therefore applies on $\Omega_q$, and every solution converges to the largest invariant set $M\subseteq\{\tilde{u}_i=0\}\cap\Omega_q$. Let a complete solution lie in that set. Then $\tilde{u}_i\equiv0$, so $u_i\equiv\bar{\omega}$ and $z_i\equiv z_i^0$ is constant, while eqn.~\eqref{eq:closed} reduces to $\dot\chi_{i,x}=\bar{\omega}\chi_{i,y}$, $\dot\chi_{i,y}=-\bar{\omega}\chi_{i,x}$, whose solution is the uniform rotation $\chi_i(t)=\Rot{-\bar{\omega} t}\chi_i(0)$. Hence $\|\chi_i(t)\|\equiv\varrho:=\|\chi_i(0)\|$ and, writing $\chi_i(0)=\varrho\,\nvec{\psi_0}$, $\tilde{u}_i(t)=\kappa_{i1}\varrho\cos(\psi_0-\bar{\omega} t)+\Pi_i'(z_i^0)\equiv0$ for all $t$.

Differentiating w.r.t. $t$ yields $\kappa_{i1}\varrho\bar{\omega}\sin(\psi_0-\bar{\omega} t)\equiv0$; since $\kappa_{i1}>0$ and $\bar{\omega}\neq0$, this forces $\varrho=0$, and the identity gives $\Pi_i'(z_i^0)=0$, i.e.\ $z_i^0=0$. Therefore $M=\{(0,0)\}$. Stability of the origin follows from positive definiteness of $V_i$ together with $\dot V_i\le0$ in the standard way, and stability with attractivity is asymptotic stability. Since $q=V_i(t_0)$ was arbitrary and $\bigcup_{q>0}\Omega_q=\mathcal{D}_i$, the region of attraction is all of $\mathcal{D}_i$.

Consider the map $\Psi:(r_{ij},e_{i_1}^j,e_{i_2}^j)\mapsto(\chi_{i,x},\chi_{i,y},z_i)$, given by eqns.~\eqref{eq:chimeas} and~\eqref{eq:zdef}. Its Jacobian with respect to $(r_{ij},e_{i_2}^j)$ has determinant $r_{ij}>0$, so $\Psi$ is a local diffeomorphism wherever the two agents are distinct. Substituting $\chi_i=0$ and $z_i=0$ into eqn.~\eqref{eq:chimeas} and solving gives $r_{ij}\cos e_{i_2}^j=R_j\sin\bar{e}_{i_1}^j$ and $r_{ij}\sin{e}_{i_2}^j=R_i-R_j\cos \bar{e}_{i_1}^j$, that is, exactly $r_{ij}=\bar r_{ij}$ and $e_{i_2}^j=\bar{e}_{i_2}^j$ of eqns.~\eqref{eq:rbar}-\eqref{eq:e2bar}. Hence convergence of $(\chi_i,z_i)$ to the origin is equivalent to convergence of $(r_{ij},e_{i_1}^j,e_{i_2}^j)$ to $(\bar r_{ij},\bar{e}_{i_1}^j,\bar{e}_{i_2}^j)$, i.e.\ $z_i\to0$ and $\zeta_i\to0$, and asymptotic stability transfers through the diffeomorphism. Furthermore $\|c_i\|=\|\chi_i\|\to0$, so eqn.~\eqref{eq:turn_centre} gives $\|P_i\|\to R_i$, and $\tilde{u}_i\to0$ gives $\dot\gamma_i=u_i\to\bar{\omega}$.

(c) From eqn.~\eqref{eq:chimeas}, $r_{ij}=g(\chi_i,z_i)$ with $g(\chi_i,z_i)=\Big\|\Big(R_j\sin e_{i_1}^j-\chi_{i,x},\ R_i-R_j\cos e_{i_1}^j-\chi_{i,y}\Big)\Big\|,$ where $e_{i_1}^j=\bar{e}_{i_1}^j+z_i$ and $g(0,0)=\bar r_{ij}$ by eqn.~\eqref{eq:rbar}. The map $g$ is $1$-Lipschitz in $\chi_i$ by the reverse triangle inequality, and $R_j$-Lipschitz in $z_i$ because $\|\frac{d}{d e_{i_1}^j}(R_j\sin e_{i_1}^j,-R_j\cos e_{i_1}^j)\|=R_j$. Hence
\setlength{\abovedisplayskip}{2pt}
\setlength{\belowdisplayskip}{2pt}
\begin{equation}
\label{eq:lip}
|r_{ij}-\bar r_{ij}|\le\|\chi_i\|+R_j|z_i| .
\end{equation}
On $\Omega_q$, eqn.~\eqref{eq:V} gives $\|\chi_i\|\le\sqrt{2R_iq/\kappa_{i1}}$ and $|z_i|\le\eta_i\sqrt{1-e^{-q/\kappa_{i2}}}$. {With $q:=V_i(T_j)<\bar q_i$, eqns.~\eqref{eq:lip} and~\eqref{eq:Gamma} give $|r_{ij}-\bar r_{ij}|\le\Gamma_i(q)<\Gamma_i(\bar q_i)=\bar r_{ij}-2R_s$ on $\Omega_q$, which is forward invariant by (a); hence $r_{ij}(t)>2R_s$ for all $t\ge T_j$.}\hfill$\square$

\subsection{Proof of Proposition~\ref{prop:graph}}
\label{app:prop}
For $t\geqslant T_c$, $c_\ell\equiv0$ and, by Lemma~\ref{lem:tc}(a), $u_\ell=\bar{\omega}$ for every leader $\ell$. Consider a follower $i$ with out-neighbour $j$. The derivation of eqn.~\eqref{eq:closed}, without using $c_j=0$, gives $\dot\chi_{i,x}=u_i\chi_{i,y}-R_i\tilde{u}_i$, $\dot\chi_{i,y}=-u_i\chi_{i,x}$ and $\dot z_i=\tilde{u}_j-\tilde{u}_i$, and Lemma~\ref{lem:tc}(c) gives $\Lambda_i=\chi_{i,x}-\langle c_j,\nvec{\gamma_i}\rangle=\chi_{i,x}-\langle\chi_j,\nvec{-e_{i_1}^j}\rangle$ in eqn.~\eqref{eq:ui}, with $e_{i_1}^j=\bar{e}_{i_1}^j+z_i$. Since $\tilde{u}_j$ depends in turn on the errors of $j$ and of its own out-neighbour, the error dynamics of follower $i$ depend only on its own errors and on those of the followers of smaller depth on its path, and the coupling terms vanish when those errors do; for a leader out-neighbour, $c_j\equiv0$, $\tilde{u}_j=0$, and eqn.~\eqref{eq:closed} is recovered. Ordering the followers by depth, the closed loop for $t\geqslant T_c$ is an autonomous, smooth cascade whose linearisation at the origin is block lower-triangular, with the matrices $A_i$ of eqn.~\eqref{eq:Ai} on the diagonal. All $A_i$ are Hurwitz, so the origin is locally exponentially stable by Lyapunov's indirect method. Finally, $\chi_i\to0$ gives $r_i\to R_i$ by eqn.~\eqref{eq:turn_centre}; $\tilde{u}_i\to0$ gives $\dot\gamma_i\to\bar{\omega}$; and, by eqn.~\eqref{eq:turn_centre}, $\Rot{-\gamma_i}(P_j-P_i)=R_j\nvec{e_{i_1}^j-\pi/2}-R_i\nvec{-\pi/2}+\Rot{-\gamma_i}(c_j-c_i)$, whose norm tends to $\bar r_{ij}$ by eqns.~\eqref{eq:body_frame_rij} and~\eqref{eq:rbar} as $z_i\to0$ and $c_i,c_j\to0$.\hfill$\square$

%

%


%

\bibliographystyle{ieeetran}        
\bibliography{autosam}

\begin{IEEEbiography}[{\includegraphics[width=1in,height=1.25in,clip,keepaspectratio]{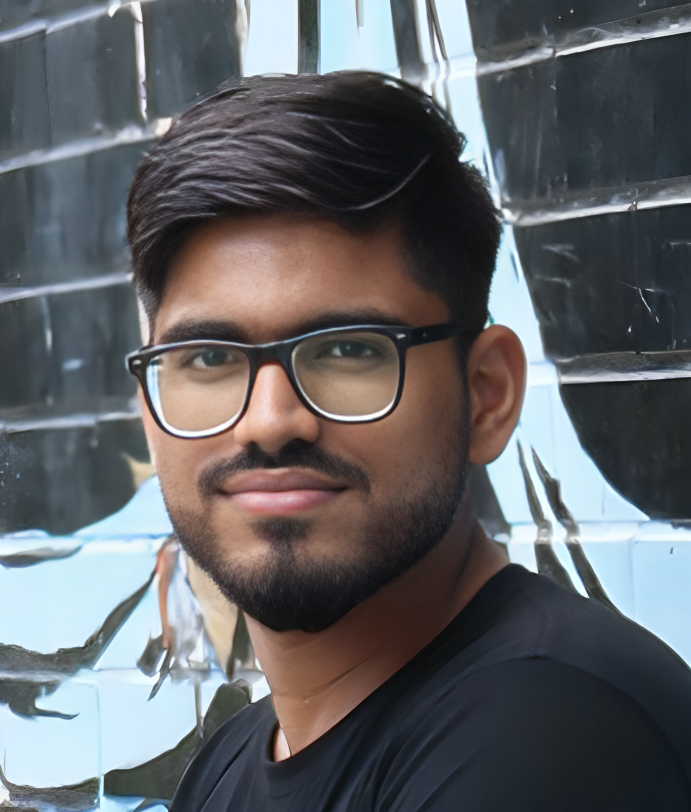}}]{Kushal P. Singh} (Memebr, IEEE) received the B.tech. degree in electrical engineering from Galgotia College of Engineering $\&$ Technology, Greater Noida, India, in 2019. He currently working toward the Ph.D. degree with the Department of Electrical Engineering, Indian Institute of Technology, Kanpur.

His research interests include path planning problems for autonomous vehicles and multi-agent systems.
\end{IEEEbiography}

\begin{IEEEbiography}[{\includegraphics[width=1in,height=1.25in,clip,keepaspectratio]{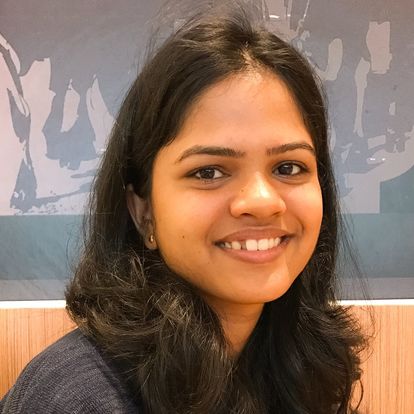}}]{Twinkle Tripathy} (Senior Memeber, IEEE) received the dual M.Tech. and Ph.D. degrees in systems and control engineering from IIT Bombay, Mumbai, India, in December 2016.

Till August 2020 she was with the school of Electrical and Electronic Engineering, NTU, Singapore, and with the faculty of Aerospace Engineering, Technion--Israel Institute of Technology, Haifa, Israel, as a Postdoctoral Fellow. She is currently an Assistant Professor with the Department of Electrical Engineering, IIT Kanpur, Kanpur, India. Her research interests include control and guidance of autonomous systems, consensus problems in multi-agent systems and opinion dynamics.
\end{IEEEbiography}

\begin{IEEEbiography}[{\includegraphics[width=1in,height=1.25in,clip,keepaspectratio]{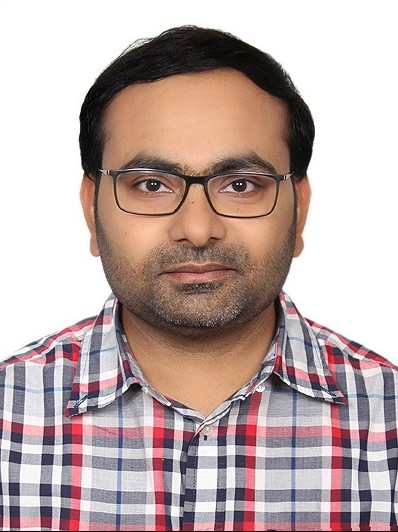}}]{Anoop Jain} (Senior Member, IEEE) is an assistant professor in the Department of Electrical Engineering at Indian Institute of Technology (IIT) Jodhpur. He received his B.Tech. (2009) from Krishna Institute of Engineering and Technology, Ghaziabad, affiliated to Uttar Pradesh Technical University. He completed his M.E. (2011) and Ph.D. (2017) at the Department of Aerospace Engineering, Indian Institute of Science (IISc), Bangalore. From 2017-2019, he was a postdoctoral fellow at the Faculty of Aerospace Engineering, Technion - Israel Institute of Technology, where his research was funded by the Lady-Davis Fellowship. His research interests include cooperative control of multi-agent systems, event-triggered control, cyber-physical systems, and nonlinear control theory.
\end{IEEEbiography}

\end{document}